\documentclass[10pt, twoside, a4paper]{article} %draft for debugging
\usepackage{amsmath}
\usepackage{graphicx}
\usepackage{epsfig}
\usepackage{amsthm}  % for theorms
\usepackage[paper=a4paper,dvips,top=3cm,left=2.5cm,right=2.5cm,foot=2cm,bottom=3.5cm]{geometry}
\usepackage{amsfonts}
\usepackage{amssymb}
\usepackage{psfrag}
\usepackage{bbm}

\newtheorem{theorem}{Theorem}
\newtheorem{lemma}{Lemma}
\newtheorem{corollary}{Corollary}
\newtheorem{definition}{Definition}
\newcommand {\ebd} {\stackrel{\Delta} {=}}

\newcommand {\reals} {{\rm I\!R}}
\newcommand {\ba} {\mbox{\boldmath $a$}}
\newcommand {\bb} {\mbox{\boldmath $b$}}

\newcommand {\bu} {\mbox{\boldmath $u$}}

\newcommand {\bx} {\mbox{\boldmath $x$}}
\newcommand {\by} {\mbox{\boldmath $y$}}
\newcommand {\bz} {\mbox{\boldmath $z$}}

\newcommand {\bX} {\mbox{\boldmath $X$}}
\newcommand {\bY} {\mbox{\boldmath $Y$}}

\newcommand{\calA}{{\cal A}}
\newcommand{\calB}{{\cal B}}
\newcommand{\calC}{{\cal C}}
\newcommand{\calD}{{\cal D}}

\newcommand{\calW}{{\cal W}}
\newcommand{\calX}{{\cal X}}
\newcommand{\calY}{{\cal Y}}
\newcommand{\calZ}{{\cal Z}}

\newcommand {\dotle} {\stackrel{\cdot} {\le}}  % leq in logarithmic scale
\newcommand {\ind} {\mathbbm{1}}
\newcommand {\sgn}{\textrm{sgn}}
\newcommand {\vol}{\textrm{Vol}}
\newcommand {\sumi} {\sum_{i=1}^n}

\newcommand {\tby} {\tilde{\by}}
\newcommand{\eqde}{\stackrel{\triangle}{=}}
\newcommand {\bs} {\boldsymbol}
\newcommand {\hrho} {\hat{\rho}}

\def\be{\begin{eqnarray}}
\def\ee{\end{eqnarray}}
\def\ben{\begin{eqnarray*}}
\def\een{\end{eqnarray*}}

\begin{document}
\title{Optimal Watermark Embedding and Detection Strategies Under Limited Detection Resources
\thanks{This research was supported by the Israel Science Foundation (grant no. 223/05).}}
\author{Neri Merhav and Erez Sabbag
}

\maketitle
\date

\begin{center}
Department of Electrical Engineering \\
Technion - Israel Institute of Technology \\
Technion City, Haifa 32000, Israel\\
{\tt \{merhav@ee, erezs@tx\}.technion.ac.il}
\end{center}
\vspace{1.6\baselineskip}
\setlength{\baselineskip}{1.4\baselineskip}
%-------------------------------------------------------------------------------------------------
\begin{abstract}
An information--theoretic approach is proposed to watermark embedding and detection under limited detector resources.
First, we consider the attack-free scenario under which asymptotically optimal decision regions in the Neyman-Pearson sense are proposed, along with the optimal embedding rule. Later, we explore the case of zero-mean i.i.d.\ Gaussian covertext distribution with unknown variance under the attack-free scenario. For this case, we propose a lower bound on the exponential decay rate of the false-negative probability and prove that the optimal embedding and detecting strategy is superior to the customary linear, additive embedding strategy in the exponential sense. Finally, these results are extended to the case of memoryless attacks and general worst case attacks. Optimal decision regions and embedding rules are offered, and the worst attack channel is identified.
\end{abstract}

%-------------------------------------------------------------------------------------------------
\section{Introduction}
\label{sec.Intro}

The field of information embedding and watermarking has become a very active field of
research in the last decade, both in the academic community and in the
industry, due to the need of protecting the vast amount of digital information
available over the Internet and other data storage media and devices (see,
e.g.,\cite{AndersonPetitcolas98}--\nocite{PAK99}\nocite{CoxMiller99}\cite{MoulinOs03}).
Watermarking (WM) is a form of embedding information secretly in a host data
set (e.g., image, audio signal, video, etc.). In this work, we raise and
examine certain fundamental questions with regard to customary methods of
embedding and detection and suggest some new ideas for the most basic setup.

Consider the system depicted in Fig.~1: Let $\bx=\big(x_1,\ldots,x_n\big)$ denote a covertext sequence emitted
from a memoryless source $P_X$, and let $\bu=\big(u_1,\ldots,u_n\big)$ denote a watermark sequence available at the embedder and at the detector.
Our work focuses on finding the optimal embedding and detection rules for the following binary hypothesis problem: under hypothesis $H_1$, the stegotext sequence $\by=\big(y_1,\ldots,y_n\big)$ is ``watermarked'' using the embedder $\by=f_n(\bx,\bu)$, while under $H_0$, $\by=\bx$, i.e, the stegotext sequence in not ``watermarked''.
An attack channel $W_n(\bz|\by)$, fed by the stegotext, produces a forgery $\bz$, which in turn, is observed by the detector. Now, given the forgery sequence $\bz$ and the watermark sequence $\bu$, the detector needs to decide whether the forgery is ``watermarked'' or not.
Performance is evaluated under the Neyman-Pearson criterion, namely, minimum false detection probability while the false alarm probability is kept lower than a prescribed level. The problem is addressed under different statistical
assumptions: the covertext distribution is known or unknown to the embedder/detector, the attack channel is known to be a memoryless attack or it is a general attack channel, and the watermark sequence is deterministic or random.

\begin{figure}[h!]
    \centering
    \psfrag{X}[][][1]{$\bx=\big(x_1,x_2,\ldots,x_n\big)$}
    \psfrag{U}[][][1]{$\bu=\big(u_1,u_2,\ldots,u_n\big)$}
    \psfrag{Y}[][][1]{$\by$}
    \psfrag{Z}[][][1]{$\bz$}
    \psfrag{W}[][][1]{$W_n(\bz|\by)$}
    \psfrag{f}[][][1]{$f_n(\bx,\bu)$}
    \psfrag{g}[][][1]{$ $}
    \psfrag{q}[][][1]{$\big\{H_0,H_1\big\}$}
    \psfrag{H0}[][][1]{$H_0$}
    \psfrag{H1}[][][1]{$H_1$}
    \includegraphics[width=6.5in]{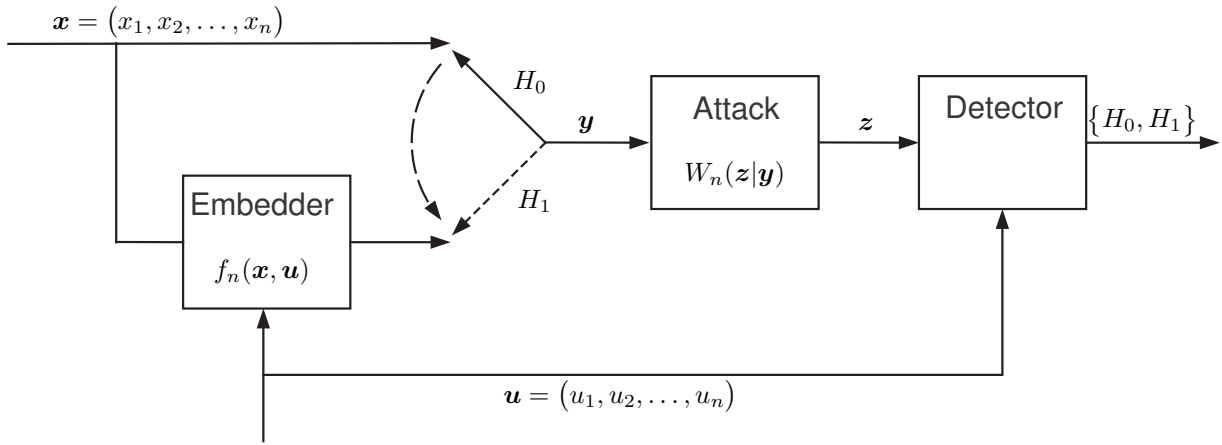}
    \label{fig_scheme}
    \caption{The watermarking and detection problem.}
\end{figure}

Surprisingly, this problem did not receive much attention in the
information theory community. In \cite{Merhav00b}, the problem of
universal detection of messages via finite state channel was considered, and
an optimal decision rule was proposed for deciding whether the observed
sequence is the product of an unknown finite-state channel fed by one of two predefined
sequences.
Liu and Moulin \cite{LiuMoulin03},\cite{LiuMoulin03b} explored the error exponent of two popular one-bit WM systems:
the spread-spectrum scheme and the quantized-index-modulation (QIM) watermarking scheme, under a general additive attack. Bounds and closed form expressions were offered for the error exponents. We note that the setting of \cite{LiuMoulin03} is different from ours: here, we are trying to find the best embedder given detection resource under Neyman-Pearson criterion of optimality, while in \cite{LiuMoulin03}, the performance (the error exponent) of a given embedding schemes and a given source distribution are evaluated under additive attacks.
In \cite{Merhav05Wacha}, the problem of embedding/detection was
formulated under limited detection resources and the optimal decision region and
the optimal embedding rule were offered to the attack-free scenario.

Many researchers from the signal/image processing community (e.g.,
\cite{PAK99},\cite{CoxMiller99},\cite{HartungKutter99}--\nocite{LinnartzKalker98}
\nocite{MillerBloom99}\nocite{MillerCoxBloom00}\cite{PodilchukDelp01},
\cite[Sec.4.2]{BarniBartolini04} and references therein) have devoted research efforts to explore the
problem of optimal watermark embedding and detection with one common
assumption: the watermark embedding rule is normally taken to be additive (linear), i.e., the
stegotext vector $\by$ is given by
\begin{equation}
\label{linear}
\by=\bx+\gamma\bu
\end{equation}
or multiplicative, where each component of $\by$ is given by
\begin{equation}
\label{mul}
y_i=x_i(1+\gamma u_i),~~~i=1,\ldots,n,
\end{equation}
where in both cases, $u_i=\pm 1$, and the choice of $\gamma$ controls the tradeoff between quality of
the stego-signal (in terms of the distortion relative to the covertext signal
$\bx$) and the detectability of the watermark - the ``signal--to--noise'' ratio.

Once the linear embedder \eqref{linear} is adopted,
elementary detection theory tells us that the optimal likelihood--ratio detector under the attack free scenario (i.e., $\bz=\by$), assuming a zero--mean, Gaussian, i.i.d.\ covertext distribution, is a correlation detector, which
decides positively ($H_1$: $\by=\bx+\gamma\bu$) if the correlation,
$\sum_{i=1}^nu_iy_i$, exceeds a certain threshold, and negatively ($H_0$:
$\by=\bx$) otherwise. The reason is that in this case, $\bx$ simply plays the
role of additive noise (the additive
embedding scheme is, in fact, the spread-spectrum modulation technique
\cite{HartungSuGirod99} in which the covertext is treated as an additive noise).
In a similar manner, the optimal test for the multiplicative embedder
(\ref{mul}) is based on the different variances of the $y_i$'s corresponding to
$u_i=+1$ relative to those corresponding to $u_i=-1$, the former being
$\sigma_x^2(1+\gamma)^2$, and the latter being $\sigma_x^2(1-\gamma)^2$, where
$\sigma_x^2$ is the variance of each component of $\bx$.

While in classical detection theory, the additivity (\ref{linear}), (or
somewhat less commonly, the multiplicativity (\ref{mul})) of the noise is part
of the channel model, and hence cannot be controlled, this is not quite the
case in watermark embedding, where one has, at least in principle, the freedom
to design an arbitrary embedding function $\by=f_n(\bx,\bu)$, trading off the
quality of $\by$ and the detectability of $\bu$. Clearly, for an arbitrary
choice of $f_n$, the above described detectors are no longer optimal in general.

Malvar and Flor{\^e}ncio \cite{MalvarFlorencio03} have noticed that better
performance can be gained if $\gamma$ is chosen as a function of the watermark
and the covertext. However, their choice does not lead to the optimal
performance as will be shown later. Recently, Furon \cite{Furon06} explored the
zero-bit watermark problem using a different setting in which the watermark
sequence is a function of the covertext and under a different criterion of
optimality.

While many papers in the literature addressed the problem of computing the
performance of different embedding and detection strategies and plotting their
receiver operating characteristics (ROC) for different values of the problem
dimension $n$ (see,
e.g.,\cite{MillerBloom99},\cite{MillerCoxBloom00},\cite{HernandezGonzalez99}
and references therein), very few works \cite{LiuMoulin03},\cite{LiuMoulin03b} deal
with the optimal asymptotic behavior of the two kinds of error probabilities,
i.e., the exponential decay rate of the two kind of the error probabilities as
$n$ tends to infinity.

The problem of finding the optimum watermark embedder $f_n$ for reliable WM
detection is not trivial: The probabilities of errors of the two kinds (false
positive and false negative) corresponding to the likelihood--ratio detector
induced by a given $f_n$, are, in general, hard to compute, and a--fortiori hard
to optimize in closed form. Moreover, obtaining closed form expressions for the optimal
embedder and decision regions when the covertext distribution is unknown is even harder
(see Section~2 for more details).

Thus, instead of striving to seek the strictly optimum embedder, we take the
following approach: Suppose that one would like to limit the complexity of the
detector by confining its decision to depend on a given set of statistics
computed from $\bz$ and $\bu$. For example, the energy of $\bz$,
$\sum_{i=1}^n z_i^2$, and the correlation $\sum_{i=1}^nu_i z_i$, which are the
sufficient statistics used by the above described correlation detector.
Other possible statistics are those corresponding to the likelihood--ratio
detector of (\ref{mul}), namely, the energies $\sum_{i:~u_i=+1}z_i^2$, and
$\sum_{i:~u_i=-1}z_i^2$, and so on. Within the class of detectors based on a given set of statistics,
we present the optimal (in the Neyman-Pearson sense) embedder and its corresponding
detector for different settings of the problem.

First, we formulate the embedding and detection problem under the attack free scenario. We devise an asymptotically optimal detector and embedding rule among all detectors which base their decisions on the empirical joint distribution of $\bz$ and $\bu$. This modeling assumption, where the detector has access to a limited set of empirical statistics of $\bu$ and $\bz$, has two motivations. First, it enables a fair comparison (in terms of detection computational resources) to different embedding/detection methods reported in the literature of WM in which most of the detectors use a similar set of statistics (mostly, correlation and energy) to base their decisions.
Second, this approach highlights the tradeoff between detection complexity and performance: Extending the set of statistics on which the detector can base its decisions, might improve the system performance, however, it increases the detector's complexity.

Later, we discuss different aspects of the basic problem, namely, practical issues regarding the implementability of the embedder, universality w.r.t. the covertext distribution, other detector's statistics, and the case where the watermark sequence is random too. These results are obtained by extending the techniques, presented in  \cite{Merhav00b},\cite{Gutman89}--\nocite{MGZ89}\cite{ZivMerhav92}, which are closely related to universal hypothesis testing problems.
We apply these results to a zero-mean i.i.d.\ Gaussian covertext distribution with unknown variance. We propose a closed-form expression for the optimal embedder, and suggest a lower bound on the false-negative probability error exponent. By analyzing the error exponent of the additive embedder and using the suggested lower bound, we show that the optimal embedder is superior to the customary additive embedder in the exponential sense.
Finally, we extend these results to memoryless attack channels and worst-case general attack channels.
The worst-attack channel is identified and optimal embedding and detection rules are offered.
The model of general worst-case attack channels, treated here, was already considered in the WM literature but in a different context. In \cite{SomekhMerhav03}, general attack channels were considered, where the capacity and random-coding error exponent where derived for the private watermarking game under general attack channels. In \cite{SomekhMerhav04}, the capacity of public watermark game under general attack channels was derived for constant composition codes. This paper is a further development and an extension of \cite{Merhav05Wacha}, \cite{SabbagMerhav06} and it gives a detailed account for the results of \cite{SabbagMerhav07}.

%-------------------------------------------------------------------------------------------------
\section{Basic Derivation}
\label{sec.Basic.Derivation}

We begin with some notation and definitions. Throughout this work, capital letters
represent scalar random variables (RVs\label{p.RV}) and specific realizations of
them are denoted by the corresponding lowercase letters. Random vectors of dimension
$n$ will be denoted by bold-face letters. The notation $\mathbbm{1}\{A\}$, where $A$
is an event, will designate the indicator function of $A$ (i.e.,$\mathbbm{1}\{A\}=1$
if $A$ occurs and $\mathbbm{1}\{A\}=0$ otherwise). We adopt the following conventions: The minimum (maximum) of a function over an empty set is understood to be $\infty$ ($-\infty$).
The notation $a_n \doteq b_n$, for two positive sequences $\{a_n\}_{n \ge 1}$ and $\{b_n\}_{n \ge 1}$, expresses
asymptotic equality in the logarithmic scale, i.e.,
$$
\lim_{n \to \infty} \frac{1}{n} \ln \left(\frac{a_n}{b_n}\right)=0.
$$
Let the vector $\hat{P}_{\bx}=\big\{\hat{P}_{\bx}(a), \; a \in \calX \big\}$ denotes the empirical distribution induced by a vector $\bx \in \calX^n$, where $\hat{P}_{\bx}(a) = \frac{1}{n}{\sum_{i=1}^n \mathbbm{1}\{x_i=a\}}$. The type class $T(\bx)$ is the set of vectors $\tilde{\bx} \in \calX^n$ such that $\hat{P}_{\tilde{\bx}}=\hat{P}_{\bx}$.
Similarly, the joint empirical distribution induced by $(\bx,\by)\in \calX^n \times \calY^n$ is the vector:
\begin{equation}
\hat{P}_{\bx\by}=\left\{\hat{P}_{\bx\by}(a,b),~a\in\calX,~b\in\calY \right\} \;,
\end{equation}
where
\begin{equation}
\hat{P}_{\bx\by}(a,b)=\frac{1}{n}\sum_{i=1}^n\mathbbm{1}\big\{x_i=a,y_i=b\big\},~~~x\in\calX,~y\in\calY \;,
\end{equation}
i.e., $\hat{P}_{\bx\by}(a,b)$ is the relative frequency of the pair $(a,b)$
along the pair sequence $(\bx,\by)$. Likewise, the type class $T(\bx,\by)$ is the set of all pairs $(\tilde{\bx},\tby) \in \calX^n \times \calY^n$ such that $\hat{P}_{\tilde{\bx}\tilde{\by}}=\hat{P}_{\bx\by}$. The conditional type class $T(\by|\bx)$, for  given vectors $\bx \in \calX^n$, and $\by \in \calY^n$ is the set of all vectors $\tby \in \calY^n$ such that $T(\bx,\tby)=T(\bx,\by)$. We denote by $\hat{E}_{\bx\by} (\cdot)$ expectation with respect to empirical joint distribution $\hat{P}_{\bx\by}$. The Kullback-Leibler divergence between two distributions $P$ and $Q$ on $\calA$, where $|\calA| < \infty$  is defined as
$$
\calD(P \| Q) = \sum_{a \in \calA} P(a) \ln \frac{P(a)}{Q(a)} \;,
$$
with the conventions that $0 \ln 0 =0$, and $p \ln \frac{p}{0}=\infty$ if $p >0$. We denote the empirical entropy of a vector $\bx \in \calX^n$ by $\hat{H}_{\bx}(X)$, where
$$
\hat{H}_{\bx}(X)= - \sum_{a \in \calX} \hat{P}_{\bx}(a) \ln \hat{P}_{\bx}(a) \;.
$$
Other information theoretic quantities governed by empirical distributions (e.g., conditional empirical entropy, empirical mutual information) will be denoted similarly.

For two vectors, $\ba,\bb \in \mathbb{R}^n$, the Euclidean inner product is defined
as $\langle\bs{a},\bs{b}\rangle= \sum_{i=1}^n a_i \cdot b_i$ and the $L_2$-norm of a
vector is defined as $\|\ba\|=\sqrt{\langle\bs{a},\bs{a}\rangle}$. Let $\vol\{A\}$ denote the volume of a set $A \subset \reals^n$, i.e., $\vol\{A\} = \int_A d\bx$. We denote by $\sgn(\cdot)$ the signum function, where $\sgn(x)=\ind\{x \ge 0\}-\ind \{x < 0 \}$.

Throughout this paper, and without essential loss of generality, we
assume that the components of $\bx$, $\by$, and $\bz$ all take on values in the same
finite alphabet $\calA$.
In Section~4, the assumption that $\calA$ is finite will be dropped, and $\calA$ will be
allowed to be an infinite set, like the real line. The components of the watermark $\bu$ will always take on values in $\calB=\{-1,+1\}$, as mentioned earlier. Let us further assume
that $\bx$ is drawn from a given memoryless source $P_X$.

Throughout the sequel, until Section~5 (exclusively), we assume that there is no attack, i.e., the channel $W_n(\bz|\by)$ is the identity channel:
$$
W_n(\bz|\by) = \left\{\begin{array}{lll} 1 &,& \bz=\by \\ 0 &,& \textrm{else} \end{array} \right. \;.
$$
This is referred to as the \emph{attack-free} scenario. In this scenario, the detector will use $\by$ and $\bu$ to base its decisions.

For a given $\bu \in \calB^n$, we would like to devise a decision rule that partitions the space $\calA^n$ of sequences $\{\by\}$, observed by the detector, into two complementary regions, $\Lambda$ and $\Lambda^c$, such that for $\by \in \Lambda$, we decide in favor of $H_1$ (watermark $\bu$ is present) and for
$\by\in\Lambda^c$, we decide in favor of $H_0$ (watermark absent: $\by=\bx$).
Consider the Neyman-Pearson criterion of minimizing the false negative
probability
\begin{equation}
\label{fn} P_{fn}=\sum_{\bx:~f_n(\bx,\bu)\in\Lambda^c}P_X(\bx)
\end{equation}
subject to the following constraints:
\begin{itemize}
\item[(1)] Given a certain distortion measure $d_e(\cdot,\cdot)$ and distortion level $D_e$, the distortion
between $\bx$ and $\by$, $d_e(\bx,\by)=d_e\big(\bx,f_n(\bx,\bu)\big)$, does not
exceed $nD_e$.
\item [(2)]
The false positive probability is upper bounded by
\begin{equation}
\label{fpc}
P_{fp}\ebd\sum_{\by\in\Lambda}P_X(\by)\le e^{-\lambda n},
\end{equation}
where $\lambda > 0$ is a prescribed constant.
\end{itemize}
In other words, we would like to choose $f_n$ and $\Lambda$ so as to minimize $P_{fn}$
subject to a distortion constraint and the constraint that the exponential decay
rate of $P_{fp}$ would be at least as large as $\lambda$.

Clearly, the problem is a classical hypothesis problem (under the Neyman-Pearson criterion of optimality), with the following hypotheses: $H_0: \by=\bx $ (the covertext is not ``marked'') and  $H_1:\by=f_n(\bx,\bu)$ (the covertext is ``marked''). Given $f_n$ and $\bu$, we can define the conditional distribution of $\by$ given the two hypotheses:
\ben
P(\by|H_0) &=& P_X(\by) \;\;, \\
P(\by|H_1) &=& \sum_{\bx: f_n(\bx,\bu)=\by} P_X(\bx) \;\;.
\een
where $P_X(\bx)$ is the covertext distribution. The optimal test which minimizes the false-negative probability under the Neyman-Pearson criterion of optimality is the likelihood ratio test (LRT) \cite[p.~34]{VanTrees68}:
\ben
L(\by)= \frac{P(\by|H_1)}{P(\by|H_0)} \begin{array}{c} H_1 \\ > \\ <
\\H_0
\end{array} \eta
\een
where $\eta$ is chosen such that
\be
P_{fp}\big(f_n,\bu\big) = \sum_{\by : L(\by) \ge  \eta} P_X(\by) = e^{-n \lambda} \;.
\ee
Note that $\eta$ is a function of $\lambda$, $f_n$ and $\bu$, therefore, we could not find a closed-form expression for $\eta$ for any general embedding rule and watermark sequence. The false-negative probability associated with the above optimal test is given by
\be
P_{fn} (f_n,\lambda,\bu) = \sum_{\by: L(\by)< \eta} \sum_{\bx:f_n(\bx,\bu)=\by}
P_X(\bx).
\ee
Now, given a distortion level $D_e$ measured using a distortion function $d_e(\cdot,\cdot)$, we would like to devise an embedder $f_n$ which minimizes the false-negative probability while the
distortion between the covertext $\bx$ and the stegotext $\by$ does not exceed
$nD_e$ and the false-positive probability is kept lower than $e^{-n \lambda}$,
i.e.,
\be
f_n^* = \arg \min_{\begin{array}{c} f_n: d_e(\bx,f_n(\bx,\bu))\le nD_e, \forall \bx\\
P_{fp}(f_n,\bu) \le e^{-n\lambda} \end{array}} P_{fn} (f_n,\lambda) \;.
\ee
The above general problem of finding the optimal embedding rule and detection
regions is by no means trivial. The fact that the probabilities of the two kinds of error cannot
be expressed in a close form make it very hard to solve this optimization
problem and, as far as we know, there is no known solution for it.
Moreover, obtaining closed form expressions for the optimal
embedder and decision regions when $P_X$ is unknown is even harder.

We therefore make an additional assumption regarding the statistics employed
by the detector. Suppose that we limit ourselves to the class of
all detectors which base their decisions on certain empirical statistics associated with $\bu$ and $\by$, for example, the empirical joint distribution of
$\by$ and $\bu$, i.e., $\hat{P}_{\bu\by}$.
Note that the requirement that the decision of the detector depends solely on $\hat{P}_{\bu\by}$ means that $\Lambda$ and $\Lambda^c$ are unions of conditional type classes of $\by$ given $\bu$.

It may seem, at a first glance, that the sequence $\bu$ is superfluous in the definition of the problem, since it is available to all legitimate parities.
However, the presence of the watermark sequence $\bu$ at the detector provides the detector with a refined version of the statistics of its input (based on the joint empirical statistics of $\by$ and $\bu$) and can be regarded as a secret key shared by both legitimate sides. This additional information at the detector improves the overall performance of the system.

For a given $\lambda>0$, define
\begin{equation}
\Lambda_*=\left\{\by:~\ln P_X(\by)+n\hat{H}_{\bu\by}(Y|U)+\lambda
n-|\calA|\ln(n+1)\leq 0\right\}.
\end{equation}
The following theorem asserts that $\Lambda_*$ is asymptotically optimal decision region:
\begin{theorem}
\begin{itemize}
\item [(i)] $P_{fp}(\Lambda_*) \le e^{-n (\lambda -\delta_n)}$ where $\lim_{n \to \infty } \delta_n=0$.
\item [(ii)] For every $\Lambda \subseteq \calA^n$ that satisfies $P_{fp}(\Lambda) \le e^{-n \lambda'}$
for some $\lambda'>\lambda$, we have $\Lambda_*^c \subseteq \Lambda^c$ for all
sufficiently large $n$.
\end{itemize}
\end{theorem}

In the above theorem it is argued that $\Lambda_*$ fulfills the false-positive constraint while minimizes the false-negative probability, i.e., for any decision region $\Lambda$ which fulfills the false-positive constraint and for any embedding rule $f_n(\bx,\bu)$ the following holds
\be
P_{fn} (\Lambda_*^c) \le P_{fn} (\Lambda^c) \;.
\ee

\begin{proof}
Let $T(\by|\bu)\subseteq \Lambda$. Then, we have
\begin{eqnarray}
e^{-\lambda n}&\geq&\sum_{\by'\in\Lambda}P_X(\by')\nonumber\\
&\geq&\sum_{\by'\in T(\by|\bu)}P_X(\by')\nonumber\\
&\geq&|T(\by|\bu)|\cdot P_X(\by)\nonumber\\
&\geq&(n+1)^{-|\calA|}e^{n\hat{H}_{\bu\by}(Y|U)}\cdot P_X(\by)\;,
\end{eqnarray}
where the first inequality is by the assumed false positive constraint, the second inequality is since
$T(\by|\bu)\subseteq \Lambda$, and the third inequality is due to the fact that
all sequences within $T(\by|\bu)$ are equiprobable under $P_X$ as they all have
the same empirical distribution, which forms the sufficient statistics for the
memoryless source $P_X$. In the fourth inequality, we use the well known lower
bound on the cardinality of a conditional type class in terms of the empirical
conditional entropy \cite{CK81}, defined as:
\begin{equation}
\hat{H}_{\bu\by}(Y|U)=-\sum_{u,y}\hat{P}_{\bu\by}(u,y)\ln\hat{P}_{\bu\by}(y|u)
\;,
\end{equation}
where $\hat{P}_{\bu\by}(y|u)$ is the empirical conditional probability of $Y$
given $U$.
We have actually shown that every $T(\by|\bu)$ in $\Lambda$ is also in
$\Lambda_*$, in other words, if $\Lambda$ satisfies the false positive
constraint (\ref{fpc}), it must be a subset of $\Lambda_*$. This means that
$\Lambda_*^c\subseteq \Lambda^c$ and so the probability of $\Lambda_*^c$ is
smaller than the probability of $\Lambda^c$, i.e., $\Lambda_*^c$ minimizes
$P_{fn}$ among all $\Lambda^c$ corresponding to detectors that satisfy
(\ref{fpc}). To establish the asymptotic optimality of $\Lambda_*$, it remains
to show that $\Lambda_*$ itself has a false positive exponent at least
$\lambda$, which is very easy to show using the techniques of \cite[eq.\
(6)]{Merhav00b} and references therein. Therefore, we will not include the
proof of this fact here. Finally, note also that $\Lambda_*$ bases its decision
solely on $\hat{P}_{\bu\by}$, as required.
\end{proof}

While this solves the problem of the optimal detector for a given $f_n$, we still
have to specify the optimal embedder $f_n^*$. Defining $\Gamma_*^c(f_n)$ to be the
inverse image of $\Lambda_*^c$ given $\bu$, i.e.,
\begin{eqnarray}
\label{revD2b}
\Gamma_*^c(f_n)&=&\Big\{\bx:~f_n(\bx,\bu)\in\Lambda_*^c\Big\}\nonumber\\
&=&\left\{\bx:~\ln P_X(f_n(\bx,\bu))+n\hat{H}_{\bu,f_n(\bx,\bu)}(Y|U)+\lambda
n-|\calA|\ln(n+1)> 0\right\},
\end{eqnarray}
then following eq.\ (\ref{fn}), $P_{fn}$ can be expressed as
\begin{equation}
\label{revD2a}
P_{fn}=\sum_{\bx\in\Gamma_*^c(f_n)}P_X(\bx).
\end{equation}
Consider now the following embedder:
\begin{equation}
\label{optemb}
f^*_n(\bx,\bu)=\mbox{argmin}_{\by:~d_e(\bx,\by)\le nD_e}\left[\ln
P_X(\by)+n\hat{H}_{\bu\by}(Y|U)\right],
\end{equation}
where ties are resolved in an arbitrary fashion. Then, it is clear by definition,
that $\Gamma_*^c(f^*_n)\subseteq\Gamma_*^c(f_n)$ for any other competing $f_n$ that
satisfies the distortion constraint, and thus $f_n^*$ minimizes $P_{fn}$ subject to
the constraints.

%-------------------------------------------------------------------------------------------------
\section{Discussion}

In this section, we pause to discuss a few important aspects of our basic results,
as well as possible modifications that might be of theoretical and practical
interest.
%-------------------------------------------------------------------------------------------------
\subsection{Implementability of the Embedder (\ref{optemb})}
\label{subsec.imp_emb}

The first impression might be that the minimization in (\ref{optemb}) is
prohibitively complex as it appears to require an exhaustive search over the
sphere $\{\by:~d_e(\bx,\by)\le nD_e\}$, whose complexity is exponential in $n$. A
closer look, however, reveals that the situation is not that bad. Note that for
a memoryless source $P_X$,
\begin{equation}
\ln P_X(\by)=-n\left[\hat{H}_{\by}(Y)+\calD(\hat{P}_{\by}\|P_X)\right] \;,
\end{equation}
where $\hat{H}_{\by}(Y)$ is the empirical entropy of $\by$ and
$\calD(\hat{P}_{\by}\|P_X)$ is the divergence between the empirical
distribution of $\by$, $\hat{P}_{\by}$, and the source $P_X$.
Moreover, if $d_e(\cdot,\cdot)$ is an additive distortion measure, i.e.,
$d_e(\bx,\by)=\sum_{i=1}^nd_e(x_i,y_i)$, then $d_e(\bx,\by)/n$ can be represented as
the expected distortion with respect to the empirical distribution of $\bx$ and
$\by$, $\hat{P}_{\bx\by}$. Thus, the minimization in \label{revD3a} (\ref{optemb}) becomes
equivalent to maximizing $[\hat{I}_{\bu\by}(U;Y)+\calD(\hat{P}_{\by}\|P_X)]$
subject to $\hat{E}_{\bx\by}d_e(X,Y)\le D_e$, where $\hat{I}_{\bu\by}(U;Y)$ denotes
the empirical mutual information induced by the joint empirical distribution
$\hat{P}_{\bu\by}$ and $\hat{E}_{\bx\by}$ denotes the aforementioned
expectation with respect to $\hat{P}_{\bx\by}$. Now, observe that for given
$\bx$ and $\bu$, both $[\hat{I}_{\bu\by}(U;Y)+\calD(\hat{P}_{\by}\|P_X)]$ and
$\hat{E}_{\bx\by}d_e(X,Y)\le D_e$ depend on $\by$ only via its conditional type
class given $(\bx,\bu)$, namely, the conditional empirical distribution
$\hat{P}_{\bu\bx\by}(y|x,u)$. Once the optimal $\hat{P}_{\bu\bx\by}(y|x,u)$ has
been found, it does not matter which vector $\by$ is chosen from the
corresponding conditional type class $T(\by|\bx,\bu)$. Therefore, the
optimization across $n$--vectors in \label{revD3b} (\ref{optemb}) boils down to optimization
over empirical conditional distributions, and since the total number of
empirical conditional distributions of $n$--vectors increases only polynomially
with $n$, the search complexity reduces from exponential to polynomial as well.
In practice, one may not perform such an exhaustive search over the discrete
set of empirical distributions, but apply an optimization procedure in the
continuous space of conditional distributions $\{P(y|x,u)\}$ (and then
approximate the solution by the closest feasible empirical distribution). At
any rate, this optimization procedure is carried out in a space of fixed
dimension, that does not grow with $n$.

%-------------------------------------------------------------------------------------------------
\subsection{Universality in the Covertext Distribution}
\label{subsec.universal}

Thus far we have assumed that the distribution $P_X$ is known. In practice,
even if it is fine to assume a certain model class, like the model of a
memoryless source, the assumption that the exact parameters of $P_X$ are known
is rather questionable. Suppose then that $P_X$ is known to be memoryless but
is otherwise unknown. How should we modify our results? First observe, that it
would then make sense to insist on the constraint (\ref{fpc}) for {\it every}
memoryless source, to be on the safe side. In other words, eq.~\eqref{fpc} would
be replaced by
\begin{equation}
\max_{P_X} \sum_{\by\in\Lambda}P_X(\by)\le e^{-\lambda n} \;,
\end{equation}
where the maximization over $P_X$ is across all memoryless sources with
alphabet $\calA$. It is then easy to see that our earlier derivation goes
through as before except that $P_X(\by)$ should be replaced by
$\max_{P_X}P_X(\by)$ in all places (see also \cite{Merhav00b}). Since
$\ln\max_{P_X}P_X(\by)=-n\hat{H}_{\by}(Y)$, this means that the modified
version of $\Lambda_*$ compares the empirical mutual information
$\hat{I}_{\bu\by}(U;Y)$ to the threshold $\lambda n-|\calA|\ln(n+1)$ (the
divergence term now disappears). By the same token, and in light of the
discussion in the previous paragraph, the modified version of the optimal
embedder (\ref{optemb}) maximizes $\hat{I}_{\bu\by}(U;Y)$ subject to the
distortion constraint. Both the embedding rule and the detection rule are then
based on the idea of {\it maximum mutual information}, which is intuitively
appealing. For more on this idea and its use as a universal decoding rule see
\cite[Sec.~2.5]{CK81}.

%-------------------------------------------------------------------------------------------------
\subsection{Other Detector Statistics}

In the previous section, we focused on the class of detectors that base their
decision on the empirical joint distribution of pairs of letters $\{(u,y)\}$.
What about classes of detectors that base their decisions on larger (and more
refined) sets of statistics? It turns out that such extensions are possible as
long as we are able to assess the cardinality of the corresponding conditional
type class. For example, suppose that the stegotext is suspected to undergo a
desynchronization attack that cyclically shifts the data by $k$ positions, where
$k$ lies in some uncertainty region, say, $\{-K,-K+1,\ldots,-1,0,1,\ldots,K\}$.
Then, it would make sense to allow the detector depend on the joint
distribution of $2K+2$ vectors: $\by$, $\bu$, and all the $2K$ corresponding
cyclic shifts of $\bu$. Our earlier analysis will carry over provided that the
above definition of $\hat{H}_{\bu\by}(Y|U)$ would be replaced the conditional
empirical entropy of $\by$ given $\bu$ and all its cyclic shifts. This is
different from the exhaustive search (ES) approach (see, e.g., \cite{Barni05})
to confront such desynchronization attacks. Note, however, that this works as
long as $K$ is fixed and does not grow with $n$.

%-------------------------------------------------------------------------------------------------
\subsection{Random Watermarks}
\label{subsec.randomWM}

Thus far, our model assumption was that $\bx$ emerges from a probabilistic
source $P_X$, whereas the watermark $\bu$ is fixed, and hence can be thought of
as being deterministic. Another possible setting assumes that $\bu$ is random
as well, in particular, being drawn from another source $P_U$, independently of
$\bx$, normally, the binary symmetric source (BSS). This situation may arise,
for example, when security is an issue and then the watermark is encrypted. In
such a case, the randomness of $\bu$ is induced by the randomness of the key.
Here, the decision regions $\Lambda$ and $\Lambda^c$ will be defined as
subsets of $\calA^n\times\calB^n$ and the probabilities of errors $P_{fn}$ and
$P_{fp}$ will be defined, of course, as the corresponding summations of
products $P_X(\bx)P_U(\bu)$. \label{revD4} The fact that $\bu$ is emitted from a memoryless
source with a known distribution, makes this model weaker compared to the model
treated above in which $\bu$ is an individual sequence. Although this model is
somewhat weaker, it can be analyzed for more general classes of detectors. This
is because the role of the conditional type class $T(\by|\bu)$ would be
replaced by the joint type class $T(\bu,\by)$, namely, the set of all {\it
pairs} of sequences $\{(\bu',\by')\}$ that have the same empirical distribution
as $(\bu,\by)$ (as opposed to the conditional type class which is defined as
the set of all such $\by$'s for a given $\bu$). Thus, the corresponding version
of $\Lambda_*$ would be
\begin{equation}
\Lambda_*=\left\{(\bu,\by):~\ln P_X(\by)+\ln
P_U(\bu)+n\hat{H}_{\bu\by}(U,Y)+\lambda n-|\calA|\ln(n+1)\le 0\right\},
\end{equation}
where $\hat{H}_{\bu\by}(U,Y)$ is the empirical joint entropy induced by
$(\bu,\by)$, and the derivation of the optimal embedder is
accordingly.\footnote{Note that in the universal case (where both $P_X$ and
$P_U$ are unknown), this leads again to the same empirical mutual information
detector as before.} The advantage of this model, albeit somewhat weaker, is
that it is easier to assess $|T(\bu,\by)|$ in more general situations than it
is for $|T(\by|\bu)|$. For example, if $\bx$ is a first order Markov source,
rather than i.i.d., and one is then naturally interested in the statistics
formed by the frequency counts of triples $\{u_i=u,~y_i=y,~y_{i-1}=y'\}$, then
there is no known expression for the cardinality of the corresponding
conditional type class, but it is still possible to assess the size of the
joint type class in terms of the empirical first-order Markov entropy of the
pairs $\{(u_i,y_i)\}$. Another example for the differences between random watermark
and deterministic watermark can be seen in Section~\ref{sec.gen.attack}.

It should be also pointed out that once $\bu$ is assumed random (say, drawn
from a BSS), it is possible to devise a decision rule that is asymptotically
optimum for an {\it individual} covertext sequence, i.e., to drop the
assumption that $\bx$ emerges from a probabilistic source of a known model. The
resulting decision rule, obtained using a similar technique, accepts $H_1$
whenever $\hat{H}_{\bu\by}(U|Y)\le 1-\lambda$, and the embedder minimizes
$\hat{H}_{\bu\by}(U|Y)$ subject to the distortion constraint accordingly.

%-------------------------------------------------------------------------------------------------
\section{Continuous Alphabet -- the Gaussian Case}

In the previous sections, we considered, for convenience, the simple case where
the components of both $\bx$ and $\by$ take on values in a finite alphabet. It
is more common and more natural, however, to model $\bx$ and $\by$ as vectors
in $\reals^n$. Beyond the fact that, summations should be replaced by
integrals, in the analysis of the previous section, this requires, in general,
an extension of the method of types \cite{CK81}, used above, to vectors with
real--valued components (see, e.g.,
\cite{Merhav89},\cite{Merhav93},\cite{MKLS94}). In a nutshell, a conditional
type class, in such a case, is the set of all $\by$--vectors in $\reals^n$
whose joint sufficient statistics with $\bu$ have (within infinitesimally small tolerance)
prescribed values, and to have a parallel analysis to that of the previous
section, we have to be able to assess the exponential order of the volume of
the conditional type class.

Suppose that $\bx$ is a zero--mean Gaussian vector whose covariance matrix is
$\sigma^2I$, $I$ being the $n\times n$ identity matrix, and $\sigma^2$ is
unknown (cf.\ Subsection~\ref{subsec.universal}). Let us suppose also that the
statistics to be employed by the detector are the energy of $\sum_{i=1}^ny_i^2$
and the correlation $\sum_{i=1}^nu_iy_i$. These assumptions are the same as in
many theoretical papers in the literature of watermark detection. Then, the
conditional empirical entropy $\hat{H}_{\bu\by}(Y|U)$ should be replaced by the
empirical differential entropy $\hat{h}_{\bu\by}(Y|U)$, given by
\cite{Merhav93}:
\be
\label{diffent}
\hat{h}_{\bu\by}(Y|U)&=& \frac{1}{2}\ln\left[2\pi e\cdot
\min_\beta\left(\frac{1}{n}\sum_{i=1}^n(y_i-\beta u_i)^2\right)\right]\nonumber\\
&=&\frac{1}{2}\ln\left[2\pi e\left(\frac{1}{n}\sum_{i=1}^ny_i^2-
\frac{(\frac{1}{n}\sum_{i=1}^nu_iy_i)^2}{\frac{1}{n}\sum_{i=1}^nu_i^2}\right)\right]\nonumber\\
&=&\frac{1}{2}\ln\left[2\pi e\left(\frac{1}{n}\sum_{i=1}^ny_i^2-
(\frac{1}{n}\sum_{i=1}^nu_iy_i)^2\right)\right].
\ee
The justification of eq.~\eqref{diffent} is as follows: For a given
$\epsilon>0$ define the set
\be
T_{\epsilon}(\by | \bu)=\left\{\tby \in \mathbb{R}^n : \Big|\sumi
y_i^2-\sumi \tilde{y}_i^2\Big| \le n \epsilon, \Big|\sumi y_i u_i-\sumi
\tilde{y}_i u_i \Big| \le n \epsilon \right\} \;.
\ee
Similarly as in  Lemma~3 \cite{Merhav93}, it can be shown
that
\be
\label{limits}
\lim_{\epsilon \to 0} \lim_{n \to \infty} \frac{1}{n} \ln  \Big[\vol \big\{ T_\epsilon
(\by|\bu) \big\} \Big] = \hat{h}_{\bu\by}(Y|U)\;.
\ee
\label{revEd}To see this, define an auxiliary channel $\by=\beta \bu+\bz$, where $\bz \sim
\mathcal{N}\big(0,\sigma_z^2 I\big)$ (this channel is used only to evaluate
$\vol \left\{ T_\epsilon (\by|\bu) \right\}$ and is not related to the actual
distribution of $\by$ given $\bu$, see \cite[p.~1262]{Merhav93}).
By tuning the parameters $\beta$ and $\sigma^2_z$ such that the expectations of $\frac{1}{n} \sumi y_i^2$ and $\frac{1}{n} \sumi y_i u_i$ would be $\frac{1}{n} \sumi \tilde{y}_i^2$ and $\frac{1}{n} \sumi \tilde{y}_i u_i$, respectively, the set $T_\epsilon(\by|\bu)$ has a high probability under the auxiliary channel given $\bu$. Moreover, any two vectors in $T_\epsilon(\by|\bu)$ have conditional pdf's which are exponentially equivalent.
Accordingly, using the same technique as in the proof of Lemma~3 in \cite[p.~1268]{Merhav93} (which is based on these observation) we derive an upper and a lower bound on $\vol \big\{T_\epsilon(\by|\bu) \big\}$. These bounds are identical in the logarithmic scale, and so,
\be
\vol \big\{ T_\epsilon (\by|\bu) \big\} 
\doteq e^{n \big[\hat{h}_{\bu\by}(Y|U)+\Delta(\epsilon)\big]} \;,
\ee
and $\lim_{\epsilon \to 0} \Delta(\epsilon) =0 $.

Note that the order in which the limits are taken in \eqref{limits} is important: We first
take the dimension $n$ to infinity, and only then we take $\epsilon$ to zero.
Mathematically speaking, if $\epsilon$ goes to zero for a finite dimension $n$
the volume of $T_\epsilon (\by|\bu)$ equals zero.
%consequently we get that $\vol \left\{ T_\epsilon (\by|\bu) \right\}=0$.
The order of the limits has a practical meaning too. The fact that $\epsilon$ is positive for any given dimension means that the detector can calculate the correlation and energy with limited precision. In the absence of such a realistic limitation, one can offer an embedding rule (under the attack-free
case and for continuous alphabet) with zero false-negative and false-positive probabilities by designing an embedder with a range having measure zero
\footnote{E.g., the spread-transform dither modulation (STDM) embedder proposed in \cite[Sec.~V.B]{ChenWornell01}
achieves zero false-negative probability under the attack-free scenario because the embedder range has measure zero. We thank M.~Barni for drawing our attention to this fact.}. This additional limitation that we implicitly impose on the detector, is very natural and it exists in every practical system.

Using the same technique used to evaluate $\hat{h}_{\bu\by}(Y|U)$ in \eqref{diffent},
it can easily be shown that
\be
\label{diff_ent_y}
\lim_{\epsilon \to 0} \lim_{n \to \infty} \frac{1}{n} \ln \Big[ \vol \{T_{\epsilon}(\by)\} \Big] =\frac{1}{2}\ln\left(2\pi e\cdot\frac{1}{n}\sum_{i=1}^ny_i^2\right)  \eqde \hat{h}_{\by}(Y) \;,
\ee
where
\be
T_\epsilon(\by)=\left\{\tilde{\by} \in \mathbb{R}^n : |\sumi y_i^2-\sumi
\tilde{y}_i^2| \le n \epsilon \right\} \; .
\ee
Therefore, the optimal embedder maximizes
\begin{equation}
\label{detector} \hat{I}_{\bu\by}(U;Y)=-\frac{1}{2}\ln\left
(1-\frac{(\frac{1}{n}\sum_{i=1}^nu_iy_i)^2}{\frac{1}{n}\sum_{i=1}^ny_i^2}\right)
\;.
\end{equation}
or, equivalently, \footnote{Note also that the corresponding detector, which
compares $\hat{I}_{\bu\by}(U;Y)$ to a threshold, is equivalent to a correlation
detector, which compares the (absolute) correlation to a threshold that depends
on the energy of $\by$, rather than a fixed threshold (see, e.g.,
\cite{Barni05}).} maximizes
\be
\label{objective}
R(\bu,\by) \eqde  \frac{\langle \bu,\by \rangle^2}{\| \by \|^2}
\ee
subject to the distortion constraint, which in this case, will naturally be
taken to be Euclidean, $\sum_{i=1}^n(x_i-y_i)^2\le nD_e$. While our discussion in
Subsection~\ref{subsec.imp_emb}, regarding optimization over conditional
distributions, does not apply directly to the continuous case considered here,
it can still be represented as optimization over a finite dimensional space
whose dimension is fixed, independently of $n$. In fact, this fixed dimension
is 2, as is implied by the next lemma.

\begin{lemma}
\label{revEb}
The optimal embedding rule under the above setting has the following form:
\begin{equation}
\label{embedder} f^*_n(\bx,\bu)=a\bx+b\bu.
\end{equation}
\end{lemma}

\begin{proof}
Clearly, every $\by \in \reals^n$ can be represented as $\by=a\bx+b\bu+\bz$,
where $a$ and $b$ are real valued coefficients and $\bz$ is orthogonal to both
$\bx$ and $\bu$ (i.e., $\langle \bu, \bz \rangle = \langle \bx, \bz \rangle
=0$). Now, for any given $\by=a\bx+b\bu+\bz$ such that $\bz \neq 0$, the vector projected onto the subspace spanned by $\bx$ and $\bu$, $\tby=a \bx + b\bu$, achieves a higher squared normalized correlation w.r.t. $\bu$ than the vector $\by$. To see this, consider the following chain of inequalities:
\be
R(\bu,\by) &=&  \frac{\langle \bu,\by \rangle^2}{\| \by \|^2}  \nonumber\\
&=& \frac{\langle \bu,a \bx + b \bu +\bz  \rangle^2}{\langle a \bx + b \bu +\bz, a \bx + b \bu +\bz \rangle} \nonumber\\
&=& \frac{\langle \bu,a \bx + b \bu \rangle^2}{\| a \bx + b \bu \|^2+ \| \bz \|^2} \nonumber\\
&\le& R(\bu,\tilde{\by}) \;.
\ee
In addition, if $\by$ fulfills the distortion constraint, then so does the projected vector $\tby$, i.e.,
\be
\|\by-\bx\|^2 &=& \|(a-1)\bx + b \bu +\bz\|^2  \nonumber\\
&=&  \|(a-1)\bx + b \bu \|^2 + \|\bz\|^2 \nonumber\\
&\ge&  \|(a-1)\bx + b \bu \|^2 \nonumber\\
&=& \|\tilde{\by}-\bx\|^2 \;.
\ee
Therefore, the optimal embedder must have the form $\by= a \bx + b \bu$. In summary, given any $\by$ that satisfies the distortion constraint, by projecting $\by$ onto the subspace spanned by $\bx$ and $\bu$, we improve the correlation without violating the distortion constraint.
\end{proof}

Upon manipulating this optimization problem, by taking advantage of its special
structure, one can further reduce its dimensionality and transform it into a
search over one parameter only (the details are in
Subsection~\ref{sec_OptEmb}).

Going back to the opening discussion in the Introduction, at first glance, this
seems to be very close to the linear embedder (\ref{linear}) that is so
customarily used (with one additional degree of freedom allowing also scaling
of $\bx$). A closer look, however, reveals that this is not quite the case
because the optimal values of $a$ and $b$ depend here on $\bx$ and $\bu$ (via
the joint statistics $\sum_{i=1}^nx_i^2$ and $\sum_{i=1}^nu_ix_i$) rather than
being fixed. Therefore, this is {\it not} a linear embedder.

%-------------------------------------------------------------------------------------------------
\subsection{Explicit Derivation of the Optimal Embedder}
\label{sec_OptEmb}

In this subsection, we present a closed-form expression for the
optimal embedder. As was shown in the previous section, the following optimization
problem should be solved:
\be
\label{opti}
& &\max \left[\frac{\left(\frac{1}{n} \sum_{i=1}^n y_i u_i\right)^2}{\frac{1}{n} \sum_{i=1}^n y_i^2} \right]. \nonumber\\
\textrm{subject to:} & & {\sum_{i=1}^n (y_i-x_i)^2 \le nD_e}
\ee
Substituting $\by=a\bx+b\bu$ in eq.~\eqref{opti}, gives:
\be
\label{opt1}
& &\max_{a,b \in \mathbb{R}} \Bigg[ \frac{a^2 \rho^2+2ab\rho+b^2}{a^2\alpha^2+2ab\rho+b^2} \Bigg] \nonumber \\
\textrm{subject to:} & & {(a-1)^2 \alpha^2+2(a-1)b \rho+b^2 \le D}
\ee
where $\alpha^2 \eqde \frac{1}{n}\sumi x_i^2$ and $\rho \eqde \frac{1}{n}\sumi
x_i u_i$. Note that $\alpha^2 \ge \rho^2$ by Cauchy-Schwarz inequality.
\begin{theorem}
\label{th1}
The optimal values of $(a,b)$ are:
\begin{itemize}
\item If $D_e \ge \alpha^2-\rho^2$:
\be
a^*=0 \qquad ; \qquad b^*=\rho+\sqrt{\rho^2-\alpha^2+D}
\ee
\item If $D_e < \alpha^2-\rho^2$:
\be
a^* &=& \arg \max \left\{t(a) \: \big| \: a \in \{a_1,a_2,a_3,a_4\} \bigcap R \right\} \nonumber\\
b^* &=& a^* \cdot t(a^*)
\ee
where
\be
t(a)&=&\frac{(1-a)\rho+\sgn(\rho)\sqrt{D_e-(a-1)^2(\alpha^2-\rho^2)}}{a}  \nonumber \\
R&=&\left[ 1-\sqrt{\frac{D_e}{\alpha^2-\rho^2}} ,1+\sqrt{\frac{D_e}{\alpha^2-\rho^2}} \:
\right] \;, \\
\ee
and
\be
a_{1,2} &=& \frac { (\alpha^2-\rho^2)(\alpha^2-D_e) \pm  \sqrt { D
{\rho}^{2}} \sqrt{(\alpha^2-\rho^2)(\alpha^2-D_e)}}{\alpha^2(\alpha^2-\rho^2)} \; \nonumber \\
a_{3,4}&=&  1 \pm \sqrt{\frac{D_e}{\alpha^2-\rho^2}} \;.
\ee
\end{itemize}
\end{theorem}
The proof is purely technical and therefore is deferred to the Appendix. We note
that in the case where $D_e \ll \alpha^2-\rho^2$, the value of $a^*$ tends to $1$, and
the value of $b^*$ tends to $\sgn(\rho)\sqrt{D_e}$. Hence, the linear embedder is not
optimal even in the case where $D_e \ll \alpha^2$. We will next use the above values
to devise a lower bound on the exponential decay rate of the false-negative
probability of the optimal embedder, and then compare it to an upper bound on the
false negative exponent of the linear embedder.

%-------------------------------------------------------------------------------------------------
\subsection{Lower Bound to the False Negative Error Exponent of the Optimal Embedder}

Since the calculation of the exact false-negative exponent of the optimal optimal embedder is highly non-trivial, in this subsection we derive a lower-bound on this exponent.
Later, we show that even this lower bound is by far larger than the exponent of the false-negative probability of the additive embedder. Therefore, the additive embedder is
sub-optimal in terms of the exponential decay rate of its false negative probability.

The lower bound will be obtained by exploring the performance of a sub-optimal embedder of the form $\by=\bx+\sgn(\rho)\sqrt{D_e} \bu$, which we name the \emph{sign embedder}. This embedder is obtained by setting $a=1$ in \eqref{embedder}(note that this value is in the allowable range $R$ of $a$). We assume that $\bX \sim \mathcal{N}\big(0,\sigma^2 I \big)$. First, we calculate a threshold value $T$
which always guarantees a false-positive exponent not smaller than $\lambda$.
Using the proposed detector \eqref{detector}, the false-positive probability
can be expressed as
\ben
P_{fp}=\Pr\Big\{\hat{I}_{\bu \by}(U;Y) > T \; \big| \; H_0 \Big\} & = &
\Pr\Big\{\hat{\rho}_{\bu \by}^2>1-e^{-2T} \; \big|\; H_0\Big\}\\
& = & 2 \Pr \Big\{\hat{\rho}_{\bu \by}> \sqrt{1-e^{-2T}}\; \big| \; H_0 \Big\}
\een
where $\hat{\rho}_{\bu \by}=\frac{\langle\bu,\by\rangle}{\|\bu\| \cdot
\|\by\|}$ is the normalized correlation between $\bu$ and $\by$. Because
under $H_0$ $\bY=\bX$, and because of the radial symmetry of the pdf of $\bX$,
we can conclude that for large $n$ \cite[p.~295]{Wyner73}:
$$
P_{fp}=\frac{2A_{n}(\theta)}{A_{n}(\pi)} \doteq e^{n \ln (\sin \theta)} \;,
$$
where $A_n(\theta)$ \footnote{It is well-known \cite[p.~293]{Wyner73} that
$A_n(\theta)=\frac{(n-1)\pi^{(n-1)/2}}{\Gamma\left(\frac{n}{2}\right)}
\int_0^\theta \sin^{(n-2)} (\varphi) d \varphi$ and $A_n(\pi)=2A_n(\pi/2)$.} is
the surface area of the $n$-dimensional spherical cap cut from a unit sphere
about the origin by a right circular cone of half angle
$\theta=\arccos\big(\sqrt{1-e^{-2T}}\big)$ ($0 < \theta \le \pi/2$). Since we
required that $P_{fp} \le e^{-n \lambda}$, then $\ln(\sin \theta)$ must not
exceed $-\lambda$, which means that
\be
- \lambda &\ge&  \ln(\sin \theta)  \nonumber \\
T &\ge& -\frac{1}{2} \ln \left[1-\cos^2 \left(\arcsin (e^{- \lambda}) \right)
\right]= \lambda \;,
\ee
where the last equality was obtained using the fact that
$\cos\big(\arcsin(x)\big)=\sqrt{1-x^2}$. Hence, setting $T=\lambda$ ensures a
false positive probability not greater than $e^{-n \lambda}$ for large $n$.
Define the false-negative exponent of the sign embedder
\be
\label{E_se_def}
E_{fn}^{se} \eqde \lim_{n \to \infty} -\frac{1}{n} \ln P_{fn}
\ee
where the false-negative probability is given by
\be
P_{fn} = \Pr\Big\{\hat{I}_{\bu \by}(U;Y) \le \lambda \;\big|\; H_1 \Big\} =
\Pr\Big\{\hat{\rho}^2_{\bu \by} \le 1-e^{-2 \lambda}\;\big|\; H_1 \Big\}.
\ee

\begin{theorem}
\label{th2}
The false-negative exponent of the sign embedder is given by
\be
\label{fn_exp_se}
E_{fn}^{se}(\lambda,D_e)=\left\{ \begin{array}{lll} 0 & , & \frac{D_e e^{-2\lambda}}{1-e^{-2\lambda}} \le \sigma^2 \\
                          \frac{1}{2}\left[\frac{D_e e^{-2\lambda}}{\sigma^2(1-e^{-2\lambda})}-\ln\left(\frac{D_e e^{-2\lambda}}{\sigma^2(1-e^{-2\lambda})}\right)-1 \right] & , & \textrm{else} \end{array}  \right.
\ee
\end{theorem}
The proof, which is mainly technical, is deferred to the Appendix. Let us
explore some of the properties of $E_{fn}^{se}(\lambda,D_e)$. First, it is clear
that $E_{fn}^{se}(0,D_e)=\infty$ (the detector output is constantly $H_1$) since
$\hrho_{\bu \by}^2 \ge 0$.
In addition, $E_{fn}^{se}(\lambda,0)=0$ ($\by=\bx$ and therefore does not contain any
information on $\bu$). For a given $D_e$, $E_{fn}^{se}(\lambda,D_e)=0$ for $\lambda
\ge \frac{1}{2} \ln \left(1+\frac{D_e}{\sigma^2}\right)$.

The exact value of the optimal exponent achieved when the optimal embedder is
employed is too involved to calculate. However, we can use some of the
properties of the optimal embedder to improve the lower bound on the optimal
exponent. According to Theorem~\ref{th1}, in the case where $D_e \ge
\alpha^2-\rho^2$, the optimal embedder can completely ``erase'' the covertext
and therefore achieves a zero false negative probability. We use this property
to improve the performance by introducing sub-optimum embedder which outperforms the sign embedder.
Since \label{revD5} $D_e \ge \alpha^2 \ge
\alpha^2-\rho^2$, the following embedding rule is obtained: $y=a\bx+b \bu$
where
\be
\label{sign_improved}
(a,b)=\left\{\begin{array}{lll} (0,\rho+\sqrt{\rho^2-\alpha^2+D_e}) & , & D_e \ge \alpha^2\\
                                (1,\sgn(\rho)\sqrt{D_e}) & , & \textrm{else}
                \end{array} \right. \;\;.
\ee
This embedder, which is an improved version of the sign embedder (but still sub-optimal), erases the
covertext in the cases where $D_e \ge \alpha^2$ (to keep the embedding rule a
function of one parameter, we chose to ``erase'' the covertext only if $D_e \ge
\alpha^2$). Its performance is presented in the following Corollary:
\begin{corollary}
\label{cor1}
For $\lambda > \frac{1}{2} \ln 2$, the false negative exponent of the improved sign
embedder is given by:
\be
\label{imp_sign_E}
\label{fn_exp_lb}
E(\lambda,D_e)=\left\{ \begin{array}{lll} 0 & , & D_e \le \sigma^2 \\
                          \frac{1}{2}\left[\frac{D_e}{\sigma^2}-\ln\left(\frac{D_e}{\sigma^2}\right)-1 \right] & , & \textrm{else} \end{array}  \right.
                          \qquad ;
\ee
otherwise, the false-negative exponent equals to $E_{fn}^{se}(\lambda,D_e)$.
\end{corollary}

The proof is deferred to the Appendix. The fact that the optimal embedder can offer
a positive false-negative exponent for every value of $\lambda$ is not surprising
due to its ability to erase the covertext, which leads to zero probability of
false-negative. Although the improved sign embedder can offer a tighter lower bound,
the improvement is made only in the case where $D_e \ge \sigma^2$ (though it is not
known a priori to the embedder). Nevertheless, it emphasizes the true potential of
the optimal embedder and the fact that the sign embedder is truly inferior to the
optimal embedder. In Figure~2, the false negative exponent of the sign embedder and
the false negative exponent of the improved embedder are plotted as functions of
$\lambda$ for a given values of $D_e$ and $\sigma$. The point where the two graphs
break apart is $\lambda=\frac{1}{2} \ln(2)$. From this point on, the improved sign embedder
achieves a fixed value of $0.5({D_e}/{\sigma^2}-\ln(D_e/\sigma^2)-1)$.

\begin{figure}[h!]
    \centering
    \psfrag{E2}[][][.7]{$E_{fn}^{se}$}
    \psfrag{E3}[][][.7]{$\qquad \; E_{fn}^{improved}$}
    \psfrag{L}[][][.65]{$\lambda$}
    \includegraphics[width=4in]{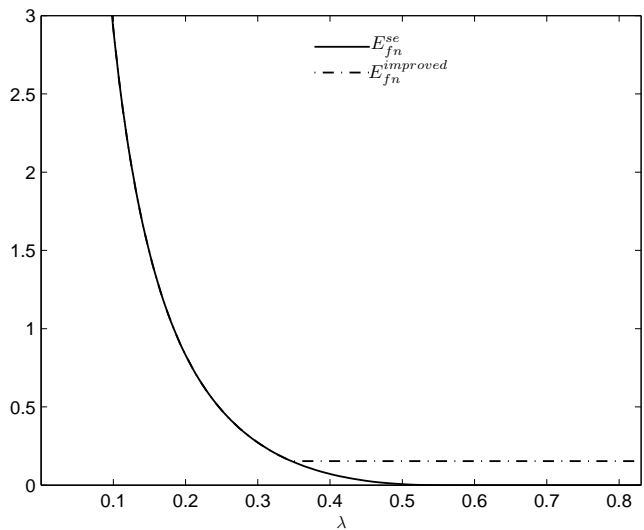}
    \label{fig1}
    \caption{Error exponents of the sign embedder and its improved version for $\sigma^2=1$ and $D_e=2$.}
\end{figure}

%-------------------------------------------------------------------------------------------------
\subsection{Comparison to the Additive Embedder}

Our next goal is to calculate the exponent of the false-negative probability of
the linear additive embedder $\by=\bx+ \sqrt{D_e}\bu$, where a normalized
correlation detector is employed. Again, we first calculate a threshold value
used by the detector which ensures a false-positive probability not greater
than $e^{-n \lambda}$. The false positive probability is given by
\be
P_{fp}= \Pr\left\{\hrho_{\bu \by} > T  \big| H_0
\right\}=\Pr\left\{\frac{\langle\bu,\bx\rangle}{\|\bu\| \cdot \|\bx\|} > T
\right\} = \frac{A_n(\theta)}{A_n(\pi)} \doteq e^{n \ln (\sin \theta)} \;,
\ee
where $\theta= \arccos(T)$ ($0 < \theta \le \pi/2$). The second equality is due to
the fact that under $H_0$ $\bY=\bX$, and the third equality is again, due to the radial
symmetry of the pdf of $\bX$. Then, $\ln (\sin \theta) \le - \lambda$ implies:
\be
T &\ge& \cos \left[\arcsin \Big(e^{-\lambda}\Big) \right] = \sqrt{1-e^{-2\lambda}} \;,
\ee
and therefore, letting $T=\sqrt{1-e^{-2\lambda}}$ ensures a false-positive
probability exponentially not greater than $e^{-n \lambda}$. Note that $\lambda \ge 0$
implies that $T$ must be non-negative. Define
\be
\label{revD7}
\Psi_1(r) &\eqde& \arccos \left[\frac{\sqrt{D_e}(T^2-1)+T
\sqrt{r-D_e(1-T^2)}}{\sqrt{r}}\right]
\ee
and
define the false-negative exponent of the additive embedder
\be
E_{fn}^{ae} \eqde \lim_{n \to \infty} -\frac{1}{n} \ln P_{fn},
\ee
where the false-negative probability is given by
\be
P_{fn}= \Pr \left\{\hrho_{\bu \by} \le  \sqrt{1-e^{-2\lambda}} \big| H_1
\right\}.
\ee

\begin{theorem}
\label{th3}
The false negative exponent of the additive embedder is given by
\be
E_{fn}^{ae}(\lambda,D_e)= \min \big\{E_1(\lambda,D_e),E_2(\lambda,D_e) \big\}
\ee
where,
\be
E_1(\lambda,D_e)& = &\min_{D_ee^{-2\lambda} < r \le \frac{D_ee^{-2\lambda}}{1-e^{-2\lambda}}} \frac{1}{2} \Bigg[\frac{r}{\sigma^2}-\ln\left(\frac{r}{\sigma^2}\right)-2\ln \sin \big(\Psi_1(r)\big)-1 \Bigg] \nonumber\\
E_2(\lambda,D_e)& = &\left\{ \begin{array}{lll} 0 & , & \frac{D_ee^{-2\lambda}}{1-e^{-2\lambda}} \le \sigma^2 \\
                          \frac{1}{2}\left[\frac{D_ee^{-2\lambda}}{(1-e^{-2\lambda}) \sigma^2}-\ln\left(\frac{D_ee^{-2\lambda}}{(1-e^{-2\lambda}) \sigma^2}\right)-1 \right] & , & \textrm{else} \end{array}  \right.
\ee
$E_{fn}^{ae}(\lambda,D_e) < E_{fn}^{se}(\lambda,D_e)$ for
$\frac{D_ee^{-2\lambda}}{1-e^{-2\lambda}} > \sigma^2$ and
\end{theorem}

Let us examine some of the properties of $E_{fn}^{ae}(\lambda,D_e)$. It is easy
to see that $E_{fn}^{ae}(\lambda,D_e) \le E_2(\lambda,D_e) =
E_{fn}^{se}(\lambda,D_e)$, i.e., the upper bound on the additive embedder
exponent serves as a lower bound on the optimal-embedder exponent. It is clear
that $E_{fn}^{ae}(\lambda,0)=0$ since $E_{fn}^{ae}(\lambda,0)\le
E_{fn}^{se}(\lambda,0)=0$. In contrast to the sign embedder, it turns out that
$E_{fn}^{ae}(0,D_e) < \infty$. To see why this is the case let us look at
\be
E_1(0,D_e)& = &\min_{r > D_e} f(r)
\ee
where $f(r)=\frac{1}{2}
\left[\frac{r}{\sigma^2}-\ln\left(\frac{r}{\sigma^2}\right)-2\ln \sin
\big(\Psi_1(r)\big)-1 \right]$. Now, since $f(r)$ is finite  for $r > D_e$, the
minimum value of $f(r)$ must be finite too. This is the case where the
threshold value equals to zero and the probability that there is an embedded
vector $\bY$ with negative correlation to $\bu$ is not zero. Clearly, for a
given $D_e$, $E_{fn}^{ae}(\lambda,D_e)=0$ for $\lambda \ge \frac{1}{2} \ln
\left(1+\frac{D_e}{\sigma^2}\right)$. Numerical calculations show that this
happens even for smaller values of $\lambda$, however, the exact smallest
value of $\lambda$ for which $E_{fn}^{ae}(\lambda,D_e)=0$ is hard to find. In
Figures~3,~4 and 5 we compare the two embedding strategies by
plotting their exponents as a functions of $\sigma^2/D_e$.

% figures
\begin{figure}[h!]
    \centering
    \psfrag{E1}[][][.65]{$E_{fn}^{ae}$}
    \psfrag{E2}[][][.65]{$E_{fn}^{se}$}
    \psfrag{L}[][][.65]{$\lambda$}
    \includegraphics[width=4in]{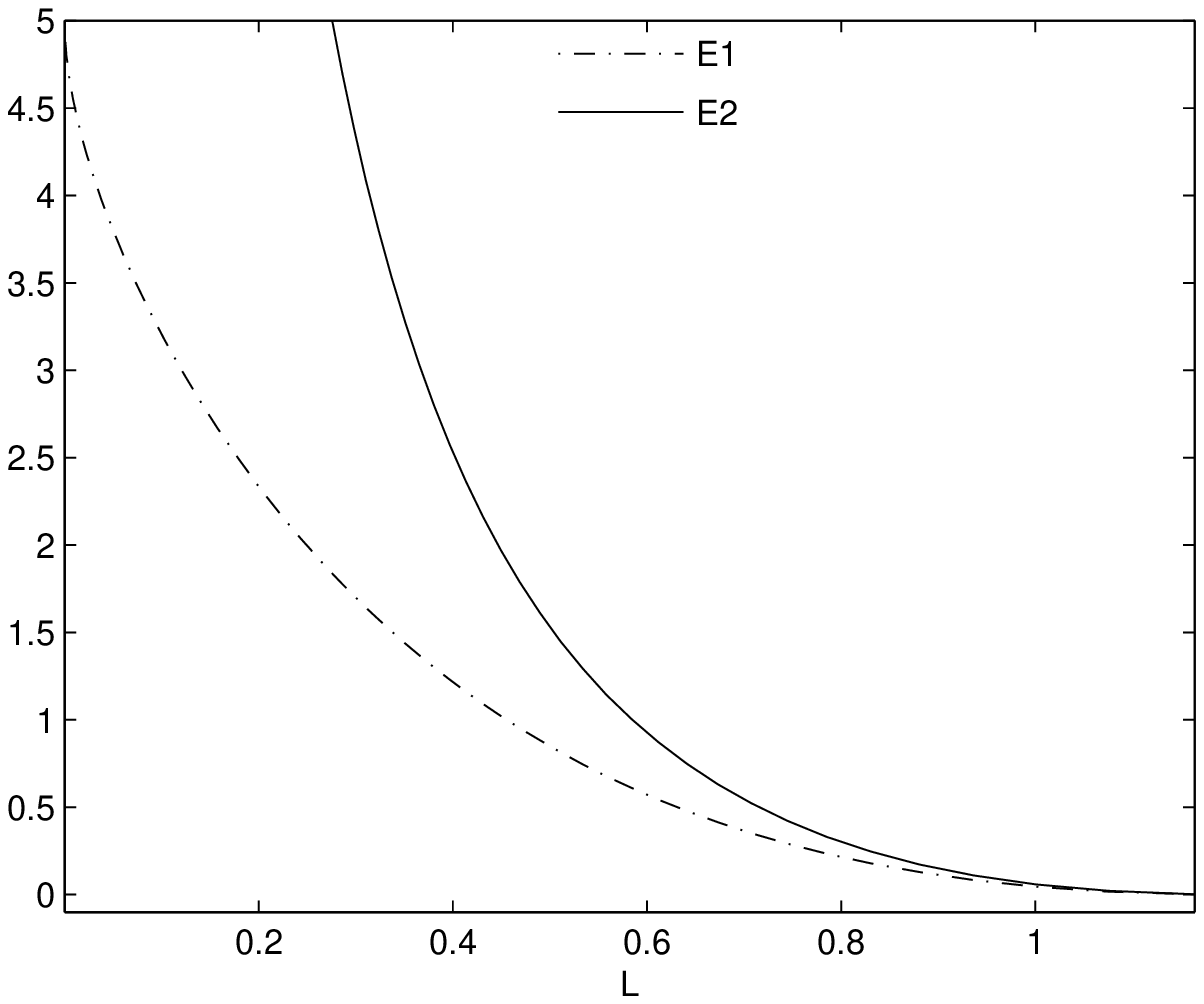}
    \label{fig2}
    \caption{Error exponents of the two embedding strategies $(\sigma^2/D=.1)$}
\end{figure}

\begin{figure}[h!]
    \centering
    \psfrag{E1}[][][.65]{$E_{fn}^{ae}$}
    \psfrag{E2}[][][.65]{$E_{fn}^{se}$}
    \psfrag{L}[][][.65]{$\lambda$}
    \includegraphics[width=4in]{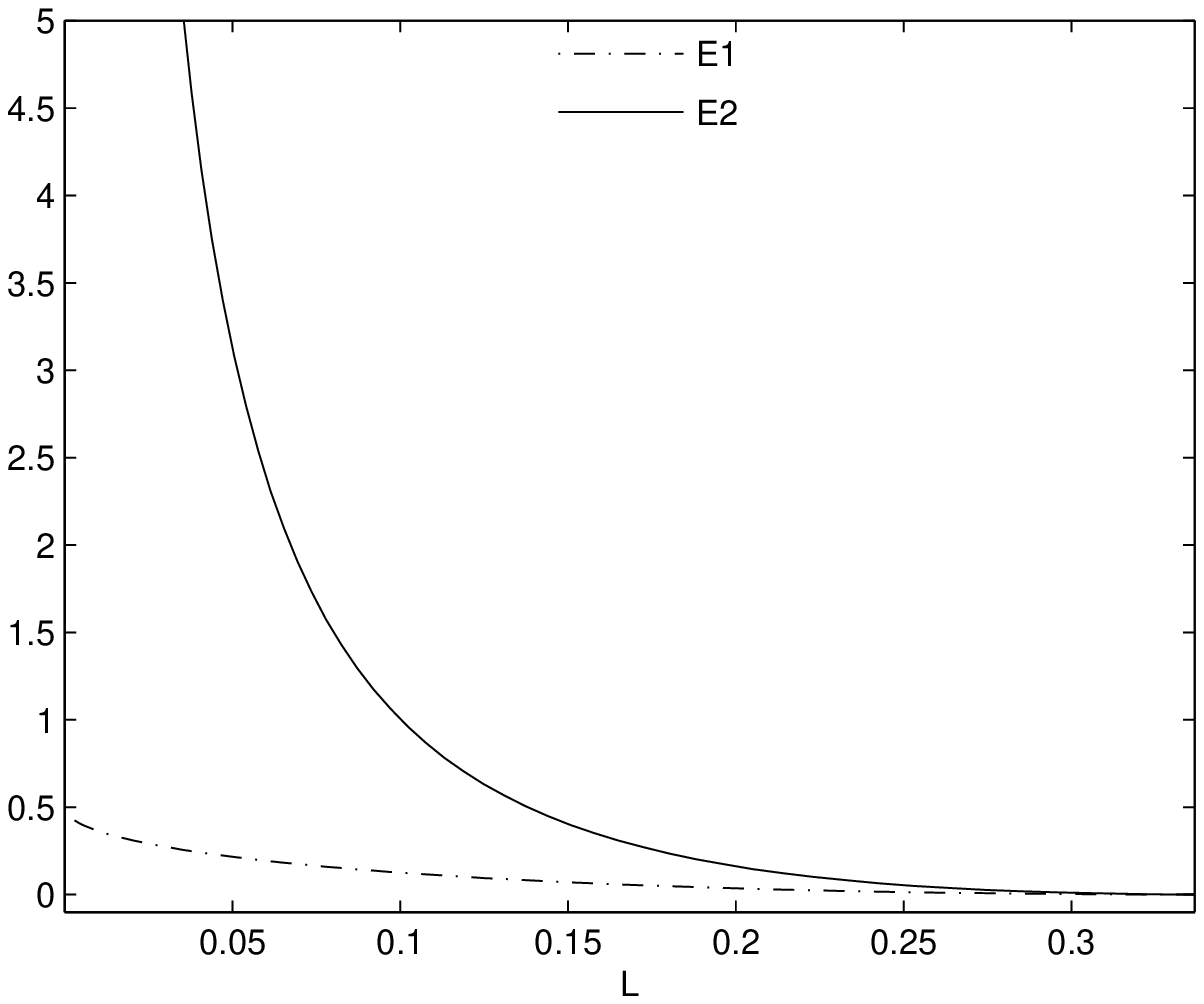}
    \label{fig3}
    \caption{Error exponents of the two embedding strategies $(\sigma^2/D=1)$}
\end{figure}

\begin{figure}[h!]
    \centering
    \psfrag{E1}[][][.65]{$E_{fn}^{ae}$}
    \psfrag{E2}[][][.65]{$E_{fn}^{se}$}
    \psfrag{L}[][][.65]{$\lambda$}
    \includegraphics[width=4in]{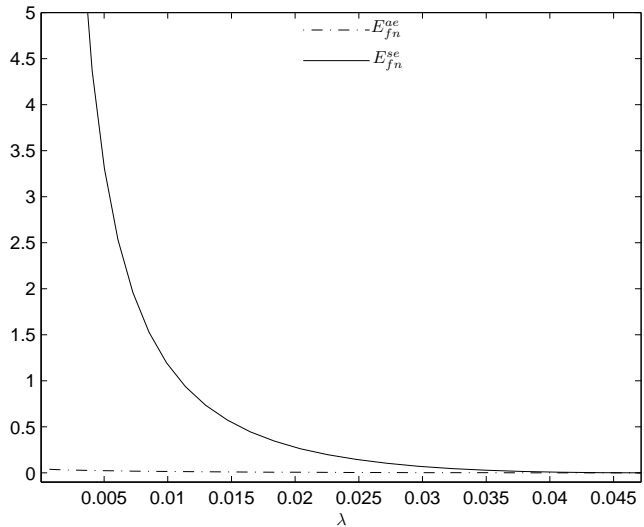}
    \label{fig4}
    \caption{Error exponents of the two embedding strategies $(\sigma^2/D=10)$}
\end{figure}

%---------------------------------------------------------------------------------------------------------------------
\subsection{Discussion}

When we take a closer look at the results, the fact the sign embedder achieves
a better performance should not surprise us. Clearly, when the correlation
between $\bx$ and $\bu$ is non-negative, the additive embedder and the sign
embedder achieve the same performance. However, when the correlation between
$\bx$ and $\bu$ is negative (this happens in probability $1/2$ due to the
radial symmetry of the pdf of the covertext) this is not true anymore. In this
case, the additive embedder tries to maximize the correlation $\rho$ between
the covertext $\bx$ and the watermark $\bu$ (while the detector compares the
normalized correlation $\hat{\rho}_{\by \bu}$ between $\by$ and $\bu$ to a
given threshold), however, these efforts are turned to the wrong direction.
Contrary to the additive embedding scheme, the sign embedder tries to maximize
the absolute value of the correlation $\rho$ while the detector compares the
absolute value of the normalized correlation to a given threshold. In this
case, the sign embedder tries to minimize the correlation $\rho$. This
difference is best exemplified in the case where $\lambda=0$. In this case, the
sign embedder achieves $E_{fn}^{se}(0,D_e)=\infty$ while $E_{fn}^{ae}(0,D_e)$ is
finite since the probability of embedded vectors $\bY$ for which
$\hat{\rho}_{\by \bu} < 0$ is not zero.

We note that although the sign embedder is suboptimal, it achieves a much
better performance than the additive embedder with a slight increase in its
complexity which is due to the calculation of $\sgn(\rho)$.

%-------------------------------------------------------------------------------------------------
\section{Attacks}
\label{sec.attacks}

Let us now extend the setup to include attacks. We first discuss attacks in
general and then confine our attention to memoryless attacks. In Section~6,
we will discuss general worst-case attacks.

The case of attack is characterized by the fact that the input to the detector
is no longer the vector $\by$ as before, but another vector,
$\bz=(z_1,\ldots,z_n)$, that is the output of a channel fed by $\by$, which we
shall denote by $W_n(\bz|\by)$ as is shown in Fig.~1. For convenience, we will assume that the
components of $\bz$ take on values in the same alphabet $\calA$, which will be assumed again to be finite, as in Sections~2 and ~3. Thus, the operation of the attack, which in general may be stochastic, is thought of as a
channel. Denoting the channel output marginal by
$Q(\bz)=\sum_{\by}P_X(\by)W_n(\bz|\by)$, the analysis of this case is, in
principle, the same as before.

Assuming, for example, that $Q$ is memoryless (which is the case when both $P_X$ and $W_n$ are memoryless, i.e., $W_n(\bz|\by)=\prod_{i=1}^n W(z_i|y_i)$ for some discrete memoryless channel $W:\calA \to \calA$),
then $\Lambda_*$ is as in Section~2, except that $P_X$, $Y$, and $\by$ should be replaced by $Q$, $Z$ and $\bz$, respectively. The optimal embedder then becomes
\begin{equation}
\label{memo}
 f^*_n(\bx,\bu)=\mbox{argmin}_{\{\by:~d_e(\bx,\by)\le nD_e\}}\sum_{\bz\in
\Lambda_*^c}W_n(\bz|\by),
\end{equation}
for the redefined version of $\Lambda_*^c$ which is given by:
\be
\Lambda_*^c &=& \left\{ \bz : \ln Q(\bz)+n \hat{H}_{\bz \bu}(Z|U)+n \lambda
-|\calA| \ln(n+1) > 0 \right\}\\
&=& \left\{ \bz : -n \hat{I}_{\bz \bu}(Z;U)-n\calD\big(\hat{P}_{\bz}\| Q
\big)+n \lambda -|\calA| \ln(n+1) > 0 \right\} \; ,
\ee
where $\hat{P}_{\bz}$ is the empirical distribution of $\bz$. Evidently, eq.~\eqref{memo}
is not a convenient formula to work with. Therefore, let us try to simplify
\eqref{memo}. For a given $\by$, let us rewrite \eqref{memo} as follows:
\be
\sum_{\bz\in \Lambda_*^c}W_n(\bz|\by)&=& \sum_{T(\bz|\by,\bu) \subseteq
\Lambda_*^c} \sum_{\bz' \in
T(\bz|\by,\bu)} W_n(\bz'|\by) \nonumber \\
&=& \sum_{T(\bz|\by,\bu) \subseteq \Lambda_*^c} \big|T(\bz|\by,\bu)\big|
W_n(\bz|\by) \;\;.
\ee
It is easy to show that for a given $\bz' \in T(\bz|\by,\bu)$ and a memoryless
channel $W_n(\bz|\by)$, the probability of $\bz'$ given $\by$ is given by the
following expression:
\be
W_n(\bz'|\by) = e^{-n\left[\hat{H}_{\by \bz}(Z|Y) + \sum_{a \in \calA}
\hat{P}_{\by}(a) \calD\big(\hat{P}_{\by\bz}(Z|Y=a)\|W(Z|Y=a)\big)\right]}.
\ee
Using the fact that the cardinality of $T(\bz|\by,\bu)$ is given by
\be
|T(\bz|\by,\bu)| \doteq e^{n \hat{H}_{\bu \by \bz}(Z|Y,U)},
\ee
we conclude that $f_n^*(\bx,\bu) \in T^*(\by|\bx,\bu)$, where $T^*(\by|\bx,\bu)$
corresponds to the following conditional empirical distribution:
\begin{multline}
\hat{P}^*_{\bu\bx\by}(Y|X,U) = \arg\max_{\begin{subarray}{l}\hat{P}_{\bu\bx\by}(Y|X,U):\\
\hat{E}_{\bx\by}d_e(X,Y)\le D_e \end{subarray}}
\;  \Bigg\{ \min_{\begin{subarray}{l}\hat{P}_{\bu\by\bz}(Z|Y,U): \\
\hat{I}_{\bu\bz}(Z;U)+\calD(\hat{P}_{\bz}\|Q) \le \lambda \end{subarray}}
\Big[\hat{I}_{\bu\by\bz}(Z;U|Y)\\+\sum_{a \in \calA} \hat{P}_{\by}(a)
\calD\big(\hat{P}_{\by\bz}(Z|Y=a) \big\| W(Z|Y=a)\big) \Big]\Bigg\}
\end{multline}
i.e., for a given $\bu$ and $\bx$, we search for the empirical distribution
$\hat{P}_{\bu\bx\by}(Y|X,U)$ which maximizes the exponent of the false negative
probability dictated by the dominating conditional type $T(\bz|\by,\bu)$ in
$\Lambda^c_*$. Once the optimal empirical distribution
$\hat{P}^*_{\bu\bx\by}(Y|X,U)$ has been found, it does not matter which vector
$\by$ is chosen from the corresponding conditional type $T^*(\by|\bx,\bu)$.

%-------------------------------------------------------------------------------------------------
\section{General Attack Channel}
\label{sec.gen.attack}

In this section we extend the results of the previous sections to include general attack channels subject to a distortion criterion.

Consider a covertext sequence $\bx=(x_1,x_2,\ldots,x_n) \in \calX^n$ emitted
from a memoryless source $P_X$ as before.
Let $d_a: \calY \times \calZ \to \reals_+$ denote another bounded single-letter
distortion measure. An attacker subject to distortion level $D_a$ w.r.t.
$d_a$ is a channel $W_n$, fed by a stegotext $\by$ and which
produces a forgery $\bz$ such that
\be
\label{attack_dist}
d_a(\by,\bz) \eqde \sum_{i=1}^n d_a(y_i,z_i) \le nD_a \quad \forall (\by,\bz) \in \calA \times \calA.
\ee
We denote the set of attack channels which satisfy \eqref{attack_dist} by
$\calW_n(D_a)$.

For a given $\bu$, we would like to devise a decision rule that partitions the
space $\calA^n$ of sequences $\{\bz\}$, observed by the detector, into two
complementary regions, $\Lambda$ and $\Lambda^c$, such that for $\bz \in
\Lambda$, we decide in favor of $H_1$ (watermark $\bu$ is present) and for $\bz
\in \Lambda^c$, we decide in favor of $H_0$ (watermark absent: $\by=\bx$).
Consider the Neyman-Pearson criterion of minimizing the worst-case false
negative probability
\begin{equation}
\label{Pfn_main}
P_{fn} \eqde \max_{W_n \in \calW_n(D_a)} P_{fn}\big(f_n,\Lambda,W_n\big)
\end{equation}
where
\be P_{fn}\big(f_n,\Lambda,W_n\big) \eqde \sum_{\bz \in \Lambda^c}
\left[\sum_{\by \in \calA^n} \left( \sum_{\bx: f_n(\bx,\bu)=\by}  P_X(\bx)
\right) W_n(\bz|\by) \right] \; ,
\ee
and $P_X(\bx)=\prod^n_{i=1}P_X(x_i)$, subject to the following constraints:
\begin{itemize}
\item[(1)] The distortion
between $\bx$ and $\by$ does not exceed $nD_e$.
\item [(2)]
The false positive probability is upper bounded by
\begin{equation}
\label{fac}
P_{fp} \ebd \max_{W_n \in \calW_n(D_a)} P_{fp}\big(\Lambda,W_n \big)   \le
e^{-n \lambda} \;,
\end{equation}
where $\lambda > 0$ is a prescribed constant and
\be
P_{fp}\big(\Lambda,W_n\big) \eqde \sum_{\bz\in\Lambda} \left(\sum_{\by \in
\calA^n} P_X(\by)W_n(\bz|\by) \right).
\ee
\end{itemize}
In other words, we would like to choose an embedder $f_n$ and a decision region
$\Lambda$ so as to minimize $P_{fn}$ subject to a distortion constraint (between
the covertext and the stegotext) and the constraint that the exponential decay
rate of $P_{fp}$ would be at least as large as $\lambda$, for \emph{any} attack
channel in $\calW_n(D_a)$.

Similarly as in Section~\ref{sec.Basic.Derivation}, we focus on the class of detectors
which base their decisions on the empirical joint distribution of $\bz$ and $\bu$.

%-------------------------------------------------------------------------------------------------
\subsection{Strongly Exchangeable Attack Channels}
\label{subsec.strongly.exchageable}

First, we restrict the set of attack channels to be strongly exchangeable channels
(the exact definition will be given in the sequel). Later, this restriction will be dropped,
and the attack channel will be allowed to be any member of $\calW_n(D_a)$. However, in this case
random watermarks (rather than deterministic ones) must be considered.

The use of strongly exchangeable channels in the context of general attack channels was proposed in \cite{SomekhMerhav03}, where Somekh-Baruch and Merhav showed (in another context) that the worst strongly exchangeable attack channel is as bad as the worst general attack channel, while strongly exchangeable channels are much easier to analyze. In the sequel, we will adjust the proof technique proposed in \cite{SomekhMerhav03} to fit our needs.

\begin{definition}
A strongly exchangeable channel $W_n$ is one that satisfies for all $\by \in
\calA^n, \bz \in \calA^n$
$$
W_n(\bz'|\by')=W_n(\bz|\by), \qquad \forall (\by',\bz') \in T(\by,\bz) \;.
$$
\end{definition}
Denote the set of all strongly exchangeable channels that operate on
$n$-tuples by $\calC_n^{ex}$ and let $\calW_n^{ex}(D_a)=\calW_n(D_a) \cap
\calC^{ex}_n$.

Define
\be
W_n^*(\bz|\by)=\frac{c_n(\by)}{|T(\bz|\by)|} \mathbbm{1} \left\{d_a(\by,\bz)
\le nD_a \right\} \;,
\ee
where, $c_n(\by)=\left[\sum_{\bz: d_a(\by,\bz) \le nD_a} \frac{1}{|T(\bz|\by)|}
\right]^{-1}$ \cite[p.~543]{SomekhMerhav03}. Clearly, $W_n^* \in \calW_{n}^{ex}(D_a)$.
Note that $c_n(\by)$ equals to the reciprocal of the number of conditional types $T(\bz|\by)$ such that $d_a(\by,\bz) \le nD_a$ \cite[p.~543]{SomekhMerhav03} which implies that  $(n+1)^{-\calA^2} \le c(\by) \le 1$. Hence, $c_n(\by)$ is at most polynomial in $n$.

Define
\be
\Lambda_*=\left\{\bz\; : \;  \hat{I}_{\bz\bu}(Z;U) + \min_{\hat{P}_{\by}:
\hat{E}_{\by \bz} d_a(Y,Z) \le D_a} \calD\big(\hat{P}_{\by} \| P_X\big) \ge
\frac{{|\calA|}\ln(n+1)}{n}+\lambda \right\}\;.
\ee

\begin{lemma}
\label{Lambda_DW}
\begin{itemize}
\item [(i)] For every $W_n \in \calW_n^{ex}(D_a)$,
$$
P_{fp}(\Lambda_*,W_n) \le e^{-n ( \lambda - \delta_n)}
$$
where $\lim_{n \to \infty} \delta_n =0$.

\item [(ii)] For any $\Lambda \subseteq \calA^n$ that satisfies
$$
P_{fp}(\Lambda,W_n) \le e^{-n \lambda'} \qquad  \forall W_n \in
\calW_n^{ex}(D_a)
$$
for some $\lambda'>\lambda$, then $\Lambda_*^c \subseteq \Lambda^c$ for all
sufficiently large $n$.
\end{itemize}
\end{lemma}

\begin{proof}
Let $T(\bz|\bu)\subseteq \Lambda$. Then, we have
\begin{eqnarray}
\label{fp_sec1}
e^{-n \lambda} &\geq& \max_{W_n \in \calW_n^{ex}(D_a)} P_{fp}(\Lambda,W_n) \nonumber \\
&=& \max_{W_n \in \calW_n^{ex}(D_a)} \sum_{\bz \in \Lambda}\left( \sum_{\by \in \calA^n}P_X(\by)W_n(\bz|\by) \right) \nonumber\\
&\geq& \sum_{\bz \in \Lambda}\left( \sum_{\by \in \calA^n}P_X(\by)W_n^*(\bz|\by) \right) \nonumber\\
&=& \sum_{T(\bz|\bu) \subseteq \Lambda} \sum_{\bz' \in T(\bz|\bu)} \left(
\sum_{\by \in \calA^n}P_X(\by)W_n^*(\bz'|\by) \right) \nonumber\\
&=& \sum_{T(\bz|\bu) \subseteq \Lambda} \sum_{\bz' \in T(\bz|\bu)} Q^*(\bz') \;,
\end{eqnarray}
where $Q^*(\bz) \eqde \sum_{\by \in \calA^n}P_X(\by)W_n^*(\bz|\by)$. Now,
\begin{eqnarray}
\label{Q_bound}
Q^*(\bz) &=& \sum_{\by \in \calA^n}P_X(\by)W_n^*(\bz|\by) \nonumber \\
&=& \sum_{T(\by|\bz) \subset \calA^n} \sum_{\by' \in T(\by|\bz)} P_X(\by')W_n^*(\bz|\by') \nonumber\\
         &=& \sum_{T(\by|\bz)\subset \calA^n} \sum_{\by' \in T(\by|\bz)} P_X(\by')\frac{c_n(\by')}{|T(\bz|\by)|} \ind\{d_a(\by',\bz) \le nD_a\} \nonumber\\
    &\doteq& \sum_{T(\by|\bz)\subset \calA^n} |T(\by|\bz)|e^{-n\left[\hat{H}_{\by}(Y)+\calD(\hat{P}_{\by}\|P_X)\right]}e^{-n\hat{H}_{\by\bz}(Z|Y)} c_n(\by)\ind\{d_a(\by,\bz) \le nD_a\} \nonumber\\
    &\doteq& \sum_{T(\by|\bz)\subset \calA^n} e^{-n\left[\hat{H}_{\by}(Y)+\calD(\hat{P}_{\by}\|P_X)-\hat{H}_{\by\bz}(Y|Z)+\hat{H}_{\by\bz}(Z|Y)\right]}c_n(\by)\ind\{d_a(\by,\bz) \le nD_a\} \nonumber\\
    &\doteq& \exp\left\{-n \left[\hat{H}_{\bz}(Z)+ \min_{\hat{P}_{\by}: \hat{E}_{\by \bz}d_a(Y,Z) \le D_a} \calD(\hat{P}_{\by} \| P_X) \right] \right\} \;,
\end{eqnarray}
where the last equality stems from the fact that $c_n(\by)$ is polynomial in $n$.

Clearly, for any $\bz' \in T(\bz)$ the following holds
\be
Q(\bz)&=& \sum_{\by \in \calA^n} P_X(\by) W_n(\bz|\by) \nonumber\\
&=& \sum_{\pi(\by)} P_X(\pi(\by)) W_n(\pi(\bz)|\pi(\by)) \nonumber\\
&=& Q(\pi(\bz)) \nonumber\\
&=& Q(\bz') \;,
\ee
where the second equality is because $W_n \in \calW_n^{ex}(D_a)$ and $\pi(\cdot)$ is a permutation of $\{1,\ldots,n\}$ such that $\bz'=\pi(\bz)$.
Hence $Q^*(\bz')=Q^*(\bz)\;\; \forall \bz' \in T(\bz)$. Following \eqref{fp_sec1}, we get
\be
e^{-n \lambda} &\ge& \sum_{T(\bz|\bu) \subseteq \Lambda} |T(\bz|\bu)| Q^*(\bz) \nonumber \\
&\ge& |T(\bz|\bu)| Q^*(\bz) \nonumber \\
&\ge& |T(\bz|\bu)| \exp\left\{-n \left[ \hat{H}_{\bz}(Z) + \min_{\hat{P}_{\by}:
\hat{E}_{\by \bz} d_a(Y,Z) \le D_a} \calD(\hat{P}_{\by} \| P_X) \right] \right\}\nonumber \\
&\ge& \exp\left\{-n\left[\hat{H}_{\bz}(Z) - n\hat{H}_{\bz\bu}(Z|U) +\min_{\hat{P}_{\by}:
\hat{E}_{\by \bz}
d_a(Y,Z) \le D_a} \calD(\hat{P}_{\by} \| P_X) \right] \right\} (n+1)^{-|\calA|} \nonumber \\
&=& \exp\left\{-n\left[\hat{I}_{\bz\bu}(Z;U) +\min_{\hat{P}_{\by}:
\hat{E}_{\by \bz}
d_a(Y,Z) \le D_a} \calD(\hat{P}_{\by} \| P_X) \right] \right\} (n+1)^{-|\calA|} \;.
\ee
In the same spirit as in the attack-free scenario, we have shown that every $T(\bz|\bu)$ in $\Lambda$ is also in $\Lambda_*$. Therefore,
$\Lambda_*^c\subseteq \Lambda^c$ and so the probability of $\Lambda_*^c$ is
smaller than the probability of $\Lambda^c$, i.e., $\Lambda_*^c$ minimizes
$P_{fn}$ among all $\Lambda^c$ corresponding to detectors that satisfy
(\ref{fac}). It remains to show
that $\Lambda_*$ itself has a false positive exponent which is at least
as large as $\lambda$ for sufficiently large $n$.

Clearly, for any attack channel $W_n \in \calW_n^{ex}(D_a)$,
\be
\label{worstch}
W_n(\bz|\by)&=& \frac{\sum_{\bz' \in T(\bz|\by)} W_n(\bz'|\by) }{|T(\bz|\by)|} \nonumber \\
&=& \frac{W\big(T(\bz|\by)|\by \big)}{|T(\bz|\by)|} \nonumber \\
&\le& \frac{1}{|T(\bz|\by)|} \ind\{d_a(\by,\bz) \le
nD_a\} \;,
\ee
where the first equality is because $W_n(\bz'|\by) = W_n(\bz|\by) \;\; \forall \bz' \in T(\bz|\by)$.
Moreover, similarly as in \eqref{Q_bound}, combined with the fact that $c(\by)$ is polynomial in $n$ implies that
\be
\label{Q_bound2}
\sum_{\by \in \calA^n} P_X(\by) \frac{\ind\big\{d_a(\by,\bz') \le nD_a\big\}}{|T(\bz'|\by)|} \doteq
\exp\left\{-n \left[\hat{H}_{\bz}(Z)+ \min_{\hat{P}_{\by}: \hat{E}_{\by \bz}d_a(Y,Z) \le D_a} \calD(\hat{P}_{\by} \| P_X) \right] \right\} \;.
\ee

Using \eqref{worstch} and \eqref{Q_bound2}, it follows that $\Lambda_*$ indeed fulfills the false-positive constraint for any attack channel $W_n \in \calW_n^{ex}(D_a)$:
\be
\max_{W_n \in \calW_n^{ex}(D_a)} P_{fp}(\Lambda_*,W_n)
&=& \max_{W_n \in \calW_n^{ex}(D_a)} \sum_{\bz \in \Lambda_*} \left( \sum_{\by \in \calA^n} P_X(\by) W_n(\bz|\by) \right) \nonumber \\
&\le& \sum_{T(\bz|\bu) \subseteq \Lambda_*} \sum_{\bz' \in T(\bz|\bu)} \left( \sum_{\by \in \calA^n} P_X(\by) \frac{\ind\big\{d_a(\by,\bz') \le nD_a\big\}}{|T(\bz'|\by)|} \right) \nonumber \\
%
%&=& \sum_{T(\bz|\bu) \subseteq \Lambda_*} \sum_{\bz' \in T(\bz|\bu)} \left[ \sum_{T(\by|\bz) \subset \calA^n} \sum_{\by' \in T(\by|\bz)} P_X(\by') \frac{\ind\big\{d_a(\by',\bz') \le nD_a\big\}}{|T(\bz'|\by')|} \right] \nonumber \\
%
&=& \sum_{T(\bz|\bu) \subseteq \Lambda_*} \sum_{\bz' \in T(\bz|\bu)} \left[ \exp\left\{-n \left(\hat{H}_{\bz}(Z)+ \min_{\hat{P}_{\by}:  \hat{E}_{\by \bz}d_a(Y,Z) \le D_a} \calD(\hat{P}_{\by} \| P_X) \right) \right\} \right] \nonumber \\
&=& \sum_{T(\bz|\bu) \subseteq \Lambda_*} e^{n \hat{H}_{\bu\bz}(Z|U)} \left[ \exp\left\{-n \left(\hat{H}_{\bz}(Z)+ \min_{\hat{P}_{\by}: \hat{E}_{\by \bz}d_a(Y,Z) \le D_a} \calD(\hat{P}_{\by} \| P_X) \right) \right\} \right] \nonumber \\
&\dotle& \sum_{T(\bz|\bu) \subseteq \Lambda_*} \exp\left\{-n\hat{I}_{\bu\bz}(Z;U)\right\} \exp\left\{-n \min_{\hat{P}_{\by}: \hat{E}_{\by \bz}d_a(Y,Z) \le D_a} \calD(\hat{P}_{\by} \| P_X)\right\} \nonumber\\
&\le& (n+1)^{|\calA|} e^{-n \lambda} \nonumber \\
&\doteq& e^{-n (\lambda-\delta_n)} \;,
\ee
where $\delta_n=\frac{|\calA| \ln(n+1)}{n} \to 0$ as $n \to \infty$.
\end{proof}

Our next step is to find an embedder which minimizes the probability of false
negative under the given decision region for any attack channels $W_n \in
\calW_n^{ex}(D_a)$. Following Section~\ref{sec.attacks}, the optimal embedder
can be written as follows:
\be
f_n^*(\bx,\bu) = \arg \min_{\by: d_e(\bx,\by)\le nD_e} \max_{W_n \in
\calW_n^{ex}(D_a)} \sum_{\bz \in \Lambda_*^c} W_n(\bz|\by) \;.
\ee

\begin{lemma}
\label{Th_opt_emb}
For any attack channel $W_n \in
\calW_n^{ex}(D_a)$, the optimal embedder $f_n^*$ which minimizes the false-negative probability
can be expressed in the following manner:
\be
\label{opt_emb1}
f_n^*(\bx,\bu)=\by, \qquad \by \in T^*(\by|\bx,\bu)
\ee
where $T^*(\by|\bx,\bu)$ corresponds to the following conditional empirical
distribution:
\begin{multline}
\label{opt_emb2}
\hat{P}_{\bu\bx\by}(Y|X,U)= \arg \max_{\begin{subarray}{l}\hat{P}_{\bu\bx\by}(Y|X,U):\\
\hat{E}_{\bx\by}d_e(X,Y)\le D_e \end{subarray}}
\;  \Bigg\{ \min_{\begin{subarray}{c}\hat{P}_{\bu\by\bz}(Z|Y,U): \\
\hat{I}_{\bu\bz}(Z;U)+\min_{\hat{P}_{\by}(Y):\hat{E}_{\by}d_a(Y,Z)\le D_a}
D(\hat{P}_{\by}\|P_X) < \lambda
\end{subarray}} \hat{I}_{\bu\by\bz}(Z;U|Y) \Bigg\} \;.
\end{multline}
\end{lemma}

\begin{proof}
For a given $\by \in \calA^n$,
\be
\max_{W_n \in \calW_n^{ex}(D_a)} \sum_{\bz \in \Lambda_*^c} W_n(\bz|\by)
&=& \max_{W_n \in \calW_n^{ex}(D_a)} \sum_{T(\bz|\by,\bu) \subseteq \Lambda_*^c} \sum_{\bz' \in T(\bz|\by,\bu)} W_n(\bz|\by) \nonumber \\
&\le& \sum_{T(\bz|\by,\bu) \subseteq \Lambda_*^c} \sum_{\bz' \in T(\bz|\by,\bu)} |T(\bz|\by)|^{-1}\ind\{d_a(\by,\bz') \le nD_a\} \nonumber \\
&\le& \sum_{T(\bz|\by,\bu) \subseteq \Lambda_*^c} |T(\bz|\by,\bu)|\cdot |T(\bz|\by)|^{-1}\ind\{d_a(\by,\bz') \le nD_a\} \nonumber \\
&\doteq& \max_{T(\bz|\by,\bu) \subseteq \Lambda_*^c} e^{-n
\hat{I}_{\bu\by\bz}(Z;U|Y)} \;.
\ee
Therefore $f_n^*(\bx,\bu) \in T^*(\by|\bx,\bu)$, where $T^*(\by|\bx,\bu)$
corresponds to the conditional empirical distribution \eqref{opt_emb2}.
\end{proof}

Note that the optimal embedder and the optimal decision rule correspond to
the case where the detector and the embedder are tuned to the worst possible channel $W_n^*$.
To extend the above results to general attack channels (i.e., channels that are members of
$\calW_n(D_a)$ rather than $\calW_n^{ex}(D_a)$) we must consider the random watermark setting (cf.\ Subsection~\ref{subsec.randomWM}). The reason for this will be made clear in the sequel.

%-------------------------------------------------------------------------------------------------
\subsection{Random Watermarks and General Attack Channels}

In the spirit of Subsection~\ref{subsec.randomWM}, from this point on, we will use the model in which $\bu$ is random as well, in particular, being drawn from another source $P_U$,
independently of $\bx$, normally, the binary symmetric source (BSS).
In this case, the decision regions $\Lambda$ and $\Lambda^c$ will be defined as
subsets of $\calA^n\times\calB^n$ and the probabilities of error $P_{fn}$ and
$P_{fp}$ will be defined, again, as the corresponding summations of
products $P_X(\bx)P_U(\bu)$.

The corresponding version of $\Lambda_*$, proposed for strongly exchangeable attack,
channels would be:
\begin{multline}
\label{Lambda_RW}
\Lambda_{**} \eqde \left\{(\bz,\bu)  : \:  \hat{I}_{\bz\bu}(Z;U) +
\calD\big(\hat{P}_{\bu}\|P_U\big) + \min_{\hat{P}_{\by}: \hat{E}_{\by \bz}
d_a(Y,Z) \le D_a} \calD\big(\hat{P}_{\by} \| P_X\big) \ge
\frac{|\calA|\ln(n+1)}{n}+\lambda \right\} \;.
\end{multline}

\begin{theorem}
\begin{itemize}
\item [(i)] For every $W_n \in \calW_n(D_a)$,
$$
P_{fp}(\Lambda_{**},W_n) \le e^{-n (\lambda-\delta_n)} \;,
$$
where $\lim_{n \to \infty} \delta_n =0$.
\item [(ii)] For any $\Lambda \subseteq \calA^n \times \calB^n$ that satisfies
$$
P_{fp}(\Lambda,W_n) \le e^{-n \lambda'} \qquad  \forall W_n \in \calW_n(D_a)
$$
for some $\lambda'>\lambda$, then $\Lambda_{**}^c \subseteq \Lambda^c$ for all
sufficiently large $n$.
\end{itemize}
\end{theorem}

To prove the above theorem in the case of general attack channels, we first
need to ensure that the probability of false positive under $\Lambda_{**}$ will
be smaller than $e^{-n \lambda}$ for any attack channel in$\calW_n(D_a)$.
We use an argument, which was used in \cite[Lemma~4]{SomekhMerhav03},
to prove that the worst strongly exchangeable attack channel
is as bad as the worst general channel, and therefore we can reuse the results of
Lemma~\ref{Lambda_DW}.
For the sake of completeness, we will rephrase the argument and adjust it to our
problem.

\begin{proof}
Given a general attack channel $W_n \in \calW_n(D_a)$, let $\pi$ denote a
permutation of $\{1,\ldots,n\}$ and let $W_n^{\pi}(\bz|\by)\eqde
W_n\big(\pi(\bz)|\pi(\by)\big)$. Clearly,
$$
\tilde{W}_n(\bz|\by)=\frac{1}{n!} \sum_{\pi} W_n^{\pi}(\bz|\by)
$$
is a strongly exchangeable channel. For a given $W_n \in \calW_n(D_a)$, let the
false-positive probability under $\Lambda$ be
\be
P_{fp}(\Lambda,W_n) &\eqde& \sum_{\bu} P_U(\bu) \sum_{\bz \in \Lambda(\bu)}
\left( \sum_{\by \in \calA^n} P_X(\by) W_n(\bz|\by) \right) \;,
\ee
where $\Lambda(\bu)=\big\{\bz\: : \: (\bz,\bu)\in \Lambda \big\}$. Recall that
any decision region $\Lambda$ is a union of joint type classes $\big\{T(\bu,\bz)\big\}$.
Since $P_{fp}(\Lambda,W_n)$ is affine in $W_n$, we can see that
\be
\label{arg1}
\frac{1}{n!} \sum_{\pi} P_{fp}(\Lambda,W_n^{\pi})&=& \frac{1}{n!} \sum_{\pi} \sum_{\bu} P_U(\bu) \sum_{\bz \in \Lambda(\bu)} \left( \sum_{\by} P_X(\by) W_n^{\pi}(\bz|\by) \right) \nonumber \\
&=& \sum_{\bu} P_U(\bu) \sum_{\bz \in \Lambda(\bu)} \left[ \sum_{\by} P_X(\by) \left(\frac{1}{n!} \sum_{\pi} W_n^{\pi}(\bz|\by) \right) \right] \nonumber \\
&=& P_{fp}\left(\Lambda,\frac{1}{n!} \sum_{\pi}W_n^{\pi} \right) \;.
\ee
Now, for a given permutation $\pi$,
\be
\label{arg2}
P_{fp}(\Lambda,W_n^{\pi})&=& \sum_{\bu} P_U(\bu) \sum_{\bz \in \Lambda(\bu)} \left(\sum_{\by} P_X(\by) W_n\big(\pi(\bz)|\pi(\by)\big) \right) \nonumber\\
&=& \sum_{\bu} P_U\big(\pi(\bu)\big) \sum_{\bz \in \Lambda(\pi(\bu))} \left( \sum_{\by} P_X\big(\pi(\by)\big) W_n\big(\pi(\bz)|\pi(\by)\big) \right) \nonumber\\
&=& \sum_{\bu} P_U(\bu) \sum_{\bz \in \Lambda(\bu)} \left( \sum_{\by} P_X(\by) W_n(\bz|\by) \right) \nonumber\\
&=& P_{fp}(\Lambda,W_n)
\ee
where the second equality follows since $\bz \in \Lambda(\bu) \Rightarrow
\pi(\bz) \in \Lambda(\pi(\bu))$ (and that is because $\Lambda$ is a union of
joint type classes $\{T(\bu,\bz)\}$) and the third equality follows from the fact that
$P_X$ and $P_U$ are memoryless which implies that
$P_X\big(\pi(\by)\big)=P_X(\by)$ and $P_U\big(\pi(\by)\big)=P_U(\by)$.

From \eqref{arg1} and \eqref{arg2}, we get that for any $\Lambda$,
\be
P_{fp}\left(\Lambda,\frac{1}{n!} \sum_{\pi}W_n^{\pi} \right) &=& \frac{1}{n!} \sum_{\pi} P_{fp}(\Lambda,W_n^{\pi}) \nonumber \\
&=& \frac{1}{n!} \sum_{\pi} P_{fp}(\Lambda,W_n) \nonumber \\
&=& P_{fp}(\Lambda,W_n) \;.
\ee
Therefore, for any $\Lambda$,
\be
\max_{W_n \in \calW_n(D_a)} P_{fp}(W_n,\Lambda) = \max_{W_n \in
\calW^{ex}_n(D_a)} P_{fp}(W_n,\Lambda) \;.
\ee
Hence, the worst general attack channel is not worse than the worst strongly exchangeable
channel, and therefore we can confine our search to the set of
strongly exchangeable channels under which $\Lambda_{**}$, defined in
\eqref{Lambda_RW}, is optimal. Using a similar proof of Lemma~\ref{Lambda_DW}, it
is easy to show that indeed under $\Lambda_{**}$ the false-positive probability
is not greater than $\exp\big\{-n (\lambda-\delta_n) \big\}$, where $\lim_{n \to \infty} \delta_n =0$.
\end{proof}

Note that the summation over $\bu$ (and the fact that any $\Lambda$ is a union
of types) enabled us the use of this argument, which might suggest that for
a deterministic watermark, a general attack channel is worse than
the worst strongly exchangeable channel. However, this channel might be dependent on the
watermark sequence which is not available to the attacker. This is exactly the reason why
random watermark setting is considered in the general attack scenario.

Once again, it is easy to verify that  $\Lambda_{**}$ does not violate the false-positive
probability constraint under general attack channel while minimizing the
false-negative probability.

We now proceed to find the optimal embedder. The false-negative probability for a given attack channel $W_n$, embedder $f_n$, and decision region $\Lambda$ can be written as follow
\be
P_{fn}\big(f_n, \Lambda, W_n\big) = \sum_{\bu \in \calB^n} P_U(\bu) P_{fn}\big(f_n, \Lambda(\bu), W_n\big),
\ee
where
\be
P_{fn}\big(f_n, \Lambda(\bu), W_n\big) = \sum_{\bz \in \Lambda^c(\bu)} \sum_{\by \in \calA^n} \left( \sum_{\bx : f_n(\bx,\bu)=\by} P_X(\bx) \right) W_n(\bz|\by) \;.
\ee

\begin{corollary}
For any attack channel $W_n \in \calW_n(D_a)$, the optimal embedder $f_n^{**}$ which minimizes the false-negative probability
is the embedder defined in (\ref{opt_emb1}).
\end{corollary}

\begin{proof}
Clearly, for any $\bu \in \calB^n$
\be
\min_{\substack{f_n(\bx,\bu): \\ d_e(\bx,\by)\le nD_e}} \max_{W_n \in \calW_n(D_a)} P_{fn}(f_n,\Lambda_{**}(\bu),W_n) %\nonumber \\
& \ge& \min_{\substack{f_n(\bx,\bu): \\ d_e(\bx,\by)\le nD_e}} \max_{W_n \in \calW_n^{ex}(D_a)} P_{fn}(f_n,\Lambda_{**}(\bu),W_n) \nonumber \\
&=& \max_{W_n \in \calW_n^{ex}(D_a)} P_{fn}(f^{*}_n,\Lambda_{**}(\bu),W_n) \;,
\ee
but on the other hand
\be
\min_{\substack{f_n(\bx,\bu): \\ d_e(\bx,\by)\le nD_e}} \max_{W_n \in \calW_n(D_a)} P_{fn}(f_n,\Lambda_{**},W_n)
& \le & \max_{W_n \in \calW_n(D_a)} P_{fn}(f^*_n,\Lambda_{**},W_n) \nonumber \\
& = & \max_{W_n \in \calW^{ex}_n(D_a)} P_{fn}(f^*_n,\Lambda_{**},W_n) \;,
\ee
where the last equality can easily be obtained from the above argument when
applied to embedders which use a certain conditional type $T(\by|\bx,\bu)$ to
produce the stegotext (as $f_n^*$). Therefore, the optimal embedder in the case
of a general attack channel is $f^*_n$, proposed in Theorem~\ref{Th_opt_emb}.
\end{proof}

Note that from \eqref{opt_emb1}, \eqref{Lambda_RW} the false-negative error exponent can be expressed in a closed form using the method of types \cite{CK81}.

%-------------------------------------------------------------------------------------------------
\subsection{Discussion}

In this section, we extended the basic setup presented
in Section~\ref{sec.Basic.Derivation} to the case of general attack channels.
First, we solved the problem for the case where the watermark sequence is
deterministic under strongly exchangeable channels. Then, we treated the case
of general attack channels, but, we had to assume that the watermark sequence
$\bu$ is random too. However, this should not surprise us. Clearly, for a given
watermark, the worst attack channel is dependent on the watermark (although it
is not known to the attacker). In this case, the attacker can imitate the
detector operation: first, it decides which hypothesis is more likely (using a
similar decision rule used by the detector). Then, it can try to ``push'' the
stegotext in the wrong direction causing a false detection. A similar behavior
can be seen in the case of a random watermark message $\bu$ and a deterministic
covertext sequence $\bx$. If $d_e=d_a$ and $D_a \ge D_e$, the worst channel
(which does depend on the covertext $\bx$) is the following: if $\by \neq \bx$
(hypothesis $H_1$) then $\bz=\bx$, i.e., the channel completely erases the
message, otherwise (hypothesis $H_0$) the channel tries to ``push'' $\by$ to
$\Lambda$. In this case, both the false-negative probability and the
false-positive probability might converge to one. The reason for that is rooted
in the fact that the set of attack channels has not been limited. In Subsection~\ref{subsec.strongly.exchageable}, we restricted the class of attack channels to be a
strongly exchangeable channel and got non-trivial results. Other limitations
may be imposed on the attack channels (e.g., blockwise memoryless, finite-state
channels) if meaningful results ought to be obtained.

Note that the worst attack strategy $W^*_n$ is independent of $\lambda$, the covertext distribution $P_X$, and even the embedder strategy and its distortion level $D_e$ (assuming that the embedder use a certain type $T(\by|\bx,\bu)$ to produce the stegotext). The attack strategy is only dependent on the allowable distortion level $D_a$.  Therefore, the embedding strategy can be designed assuming that the worst attack channel is present. This can be useful in evaluating the performance (in terms of false-negative probability) of suboptimal embedders.

%-------------------------------------------------------------------------------------------------
\renewcommand{\theequation}{A-\arabic{equation}} % redefine the command that creates the equation no.
\setcounter{equation}{0}  % reset counter

\section*{Appendix}  % use *-form to suppress numbering

\begin{proof}[Proof of Theorem~\ref{th1}]
First, we explore the case where $a=0$, i.e., $\by=b\bu$. Substituting $a=0$ in
the constraint of eq.~\eqref{opt1}, we get that $b^2-2 \rho b+(\alpha^2-D_e) \le
0$. The fact that $b$ is a real number implies that the discriminant of $(b^2-2
\rho b+(\alpha^2-D_e))$ is non-negative which leads to $\rho^2-(\alpha^2-D_e) \ge
0$, or $ D \ge \alpha^2-\rho^2$. This corresponds to the case where the
stegotext includes \emph{only} a fraction of $\bu$ without violating the
distortion constraint. In this case, the false-negative probability is zero
(the distortion constraint is so loose, it allows to ``erase'' the covertext).
In the following case, we can choose $b^*=\rho+\sqrt{\rho^2-\alpha^2+D}$ as the
optimal solution. From now on, we assume that $D_e < \alpha^2-\rho^2$ which means
that $a=0$ is not a legitimate solution. Let us assume that $\rho \ge 0$.
Define $t \eqde b/a$, and rewrite \eqref{opt1} by dividing the numerator and
denominator by $a^2$:
\be
& & \max_{t \in \mathbb{R}} f(t) \nonumber\\
\textrm{subject to:} & &  a^2 t^2+2(a-1)a\rho t+(a-1)^2 \alpha^2  \le D
\ee
where
$$
f(t)=\frac{(t+\rho)^2}{(t+\rho)^2+(\alpha^2-\rho^2)} \;.
$$
It is easy to show that maximizing $f(t)$ is equivalent to maximizing $t$. Since $t$
is a real number, the discriminant of $\big[a^2 t^2+2(a-1)a\rho t+(a-1)^2
\alpha^2-D\big]$ must be non-negative, i.e.,
\be
\Delta=4a^2 \left[D-(a-1)^2(\alpha^2-\rho^2)\right] \ge 0 \;\;,
\ee
which leads to
\be
1-\sqrt{\frac{D_e}{\alpha^2-\rho^2}} \le a \le 1+\sqrt{\frac{D_e}{\alpha^2-\rho^2}}.
\ee
Hence, $a$ must be in the range $R \eqde \left[\: 1-\sqrt{\frac{D_e}{\alpha^2-\rho^2}}
,1+\sqrt{\frac{D_e}{\alpha^2-\rho^2}}\: \right]$.
Let us rewrite the constraint as follows,
\be
\left[at+(a-1)\rho \right]^2+(a-1)^2(\alpha^2-\rho^2)-D \le 0 \;\;,
\ee
consequently,
\be
\frac{(1-a)\rho - \sqrt{D_e-(a-1)^2 (\alpha^2-\rho^2)}}{a} \le t \le \frac{(1-a)\rho +
\sqrt{D_e-(a-1)^2 (\alpha^2-\rho^2)}}{a} \;\;.
\ee
Our next step will be to maximize the upper bound on $t$ in the allowable range of
$a$.
\be
\arg \max_{a \in R} \: t(a)
\ee
where
\be
t(a) = \frac{(1-a)\rho + \sqrt{D_e-(a-1)^2 (\alpha^2-\rho^2)}}{a} \;.
\ee
After differentiating with respect to $a$ and equating to zero, we get
\be
a_{1,2}= \frac { (\alpha^2-\rho^2)(\alpha^2-D_e) \pm  \sqrt { D {\rho}^{2}}
\sqrt{(\alpha^2-\rho^2)(\alpha^2-D_e)}}{\alpha^2(\alpha^2-\rho^2)} \;.
\ee
Accordingly, the optimal value of $a$ and $b$ are
\be
\label{otp_emb}
(a^*,b^*)=\left(\arg \max \left\{t(a) | a \in \{a_1,a_2,a_3,a_4\} \bigcap R
\right\},a^* \cdot t(a^*) \right) \;\;,
\ee
where $a_{3,4}=1 \pm \sqrt{\frac{D_e}{\alpha^2-\rho^2}}$. The same results are
obtained in the case where $\rho <0$.
\end{proof}

\begin{proof}[Proof of Theorem~\ref{th2}]
It is easy to show that under $H_1$
\be
\label{rho}
\hat{\rho}^2_{\bu \by}=
\frac{\left(|\rho|+\sqrt{D_e}\right)^2}{\left(|\rho|+\sqrt{D_e}\right)^2+(\alpha^2-\rho^2)}
\; ,
\ee
where $\alpha^2$ and $\rho$ are functions of the random vector $\bX$. By
conditioning on $\alpha^2$, we can express the false-negative probability as
\be
P_{fn}= \int_0^{\infty}  \Pr \left\{ \hat{\rho}^2_{\bu \by} \le 1-e^{-2\lambda}
\; \Big|\; H_1, \alpha^2=r \right\} \cdot p_{\alpha^2}(r) dr,
\ee
where $(n\alpha^2/\sigma^2)$ is $\chi^2$ distributed with $n$ degrees of freedom and
the probability density function for the $\chi^2$ distribution with $n$ degrees of
freedom is given by
$$
p_{\chi^2_n}(z)= \frac{(1/2)^{n/2}}{\Gamma(n/2)}{z}^{n/2-1} e^{-n/2}, \qquad z
\ge 0 \quad ,
$$
and $\Gamma(\cdot)$ denotes the Gamma function.
Now, given $\alpha^2$, $D_e$ and a threshold value $\tau \eqde 1-e^{-2\lambda}$,
let us find the range of $\rho$ for which $\hat{\rho}^2_{\bu \by} \le \tau$, i.e.,
\be
\label{equation}
\hat{\rho}^2_{\bu \by}(\rho) \eqde \frac{(|\rho|+\sqrt{D_e})^2}{(|\rho|+\sqrt{D_e})^2+(\alpha^2-\rho^2)} \le \tau \;\;.
\ee
The function $\hat{\rho}^2_{\bu \by}(\rho)$ is symmetric with respect to the
$\rho$ axis, monotonically increasing in $|\rho|$ and attains its minimum value
$\frac{D_e}{D_e+\alpha^2} $ at $\rho=0$. Hence, for $\alpha^2 <
\frac{D_e(1-\tau)}{\tau}$, $\hat{\rho}_{{\bu \by}}^2$ is greater than $\tau$.
After solving \eqref{equation} with respect to $\rho$ and using the fact that $\tau
\le 1$, we get that $|\hat{\rho}_{\bu \by}| \le \sqrt{\tau}$ implies that
$|\rho| \le \sqrt{D_e} (\tau-1)+\sqrt{D_e \tau^2+\tau \alpha^2-\tau D}$ as long as
$\alpha^2 \ge \frac{D_e(1-\tau)}{\tau}$. Define
\be
\Theta(r) \eqde \arccos \left[\frac{\sqrt{D_e}(\tau-1)+\sqrt{D_e \tau^2+\tau r -\tau
D}}{\sqrt{r}} \right]
\ee
It follows that
\ben
\Pr \left\{ \hat{\rho}^2_{\bu \by} \le \tau \; \Big|\; H_1, \alpha^2=r \right\}
&=& \Pr \left\{ \rho^2 \le \left[\sqrt{D_e} (\tau-1)+\sqrt{D_e \tau^2+\tau
\alpha^2-\tau D}\right]^2\; \Big|\; H_1, \alpha^2=r \right\} \\
 &=& 1-\Pr \left\{ \rho^2 > \left[\sqrt{D_e} (\tau-1)+\sqrt{D_e \tau^2+\tau
\alpha^2-\tau D}\right]^2\; \Big|\; H_1, \alpha^2=r \right\} \\
&=& 1-2\frac{A_n\big(\Theta(r)\big)}{A_n(\pi)} \;\;,
\een
where
$$
\frac{A_n\big(\Theta(r)\big)}{A_n(\pi)} \doteq e^{n \ln \sin\big(\Theta(r)\big)} \;.
$$
We note that $\Pr \left\{ \hat{\rho}^2_{\bu \by} \le \tau  \big| H_1,
\alpha^2\right\}=0$ for $\alpha^2$ in the range
$\left[0,\frac{D_e(1-\tau)}{\tau}\right]$. Therefore,
\be
\label{int_se}
P^{(n)}_{fn} & = & \frac{(1/2)^{n/2}}{\Gamma(n/2)}
\int_{\frac{D_e(1-\tau)}{\tau}}^\infty
\left[1- e^{n \ln {\sin\big(\Theta(r)\big)}}\right] e^{-\frac{nr}{2 \sigma^2}} \Bigg(\frac{nr}{\sigma^2} \Bigg)^{\frac{n-2}{2}}dr \nonumber\\
& = & \frac{(1/2)^{\frac{n}{2}} n^{\frac{n-2}{2}}}{\Gamma(n/2)}
\left[\int_{\frac{D_e(1-\tau)}{\tau}}^\infty \frac{\sigma^2}{r} e^{-\frac{nr}{2
\sigma^2}} e^{\frac{n}{2} \ln(r/\sigma^2)} dr - \int_{\frac{D_e(1-\tau)}{\tau}}^\infty
e^{n \ln {\sin\Theta(r)}}e^{-\frac{nr}{2 \sigma^2}} e^{\frac{n}{2} \ln(r/\sigma^2)}
dr \right] \;. \nonumber\\
\ee
Our next step is to evaluate the exponential decay rate of \eqref{int_se}. It
is easy to see that the first integral of \eqref{int_se} has a slower
exponential decay rate and therefore dictates the overall decay rate. To
evaluate the exponential decay rate of $P^{(n)}_{fn}$ as $n \to \infty$ we use
Laplace's method for integrals
\footnote{ \label{revD8}
Laplace's method is a general
technique for obtaining the asymptotic behavior of integrals of the form $I(x) = \int_a^b f(t)
e^{x \Phi(t)} dt$ as $x \to \infty$.
In this case $c \in [a,b]$, the maximum of $\Phi(t)$ in the
interval $[a,b]$, dictates the asymptotic behavior of the integral (assuming
that $f(c) \neq 0$), or in the above case:
$$
\lim_{n \to \infty} -\frac{1}{n} \ln \left[ \int_a^b f(t) e^{-n \Phi(t)} dt
\right] = \min_{t \in [a,b]} \Phi(t) \;\;.
$$
See \cite[Sec.~6.4]{BenderOrszag78}, \cite[Ch.4]{DeBruijn70} for more information.}.
Therefore, we need to find the slowest exponential decay rate of the integrant in the limits
of the integral.
It is easy to show that
\be
\lim_{n \to \infty} \frac{1}{n} \ln \left[\frac{(1/2)^{\frac{n}{2}}
n^{\frac{n-2}{2}}}{\Gamma(n/2)}\right] = \frac{1}{2} \;\;,
\ee
and therefore the overall exponent is given by
\be
\label{E_se_min}
E_{fn}^{se}(\tau,D_e)= \min_{r \ge \frac{D_e(1-\tau)}{\tau}}
\frac{1}{2}\Bigg[\frac{r}{\sigma^2}-\ln(r/\sigma^2)-1\Bigg].
\ee
\label{revD9}
The function $g(r) = \big[r/\sigma^2-\ln(r/\sigma^2)-1 \big], \; r \in
(0,\infty)$ achieves its minimum at $r=\sigma^2$ and $g(\sigma^2)=0$.
Therefore, in the case where $\frac{D_e(1-\tau)}{\tau} \le \sigma^2$,
$E_{fn}^{se}(\tau,D_e)=0$. Otherwise,
the minimum of \eqref{E_se_min} is obtained at $r=\frac{D_e(1-\tau)}{\tau}$.
Hence, the false-negative exponent of the sign embedder is given by
\be
E_{fn}^{se}(\tau,D_e)=\left\{ \begin{array}{lll} 0 & , & \frac{D_e(1-\tau)}{\tau} \le \sigma^2 \\
                          \frac{1}{2}\left[\frac{D_e(1-\tau)}{\tau \sigma^2}-\ln\left(\frac{D_e(1-\tau)}{\tau \sigma^2}\right)-1 \right] & , & \textrm{else} \end{array}  \right.
\ee
Setting $\tau=1-e^{-2\lambda}$ achieves \eqref{fn_exp_se}.

\end{proof}

\begin{proof}[Proof of Corollary~\ref{cor1}]
Since the false-negative probability of the improved embedder \eqref{sign_improved}
is zero for $\alpha^2 \le D_e$ we can rewrite the integral \eqref{int_se} for the case
where $\frac{1-\tau}{\tau} \le 1$ (or $\lambda \ge 1/2 \ln 2$) where the lower limit
equals to $D_e$ (and does not depend on $\lambda$) as following:
\be
\label{int_imp_se}
P^{(n)}_{fn} & = & \frac{(1/2)^{n/2}}{\Gamma(n/2)} \int_{D_e}^\infty \left[1- e^{n \ln
{\sin\big(\Theta(r)\big)}}\right] e^{-\frac{nr}{2 \sigma^2}}
\Bigg(\frac{nr}{\sigma^2} \Bigg)^{\frac{n-2}{2}}dr \;\;.
\ee
Optimizing using Laplace method as done in the proof of Theorem~\ref{th2} leads to
\eqref{imp_sign_E}.
\end{proof}

\begin{proof}[Proof of Theorem~\ref{th3}]
Given $\lambda>0$, the false-negative probability is given by
\be
P_{fn}= \Pr \left\{\hrho_{\bu \by} \le  \sqrt{1-e^{-2\lambda}} \big| H_1
\right\} \;,
\ee
where the normalized correlation, under $H_1$, is given by
\be
\hrho_{\bu \by} = \frac{\rho+\sqrt{D_e}}{\sqrt{\alpha^2+2\sqrt{D_e}\rho+D}} &<& T
\;.
\ee
The function $\hrho_{\bu \by}(\rho)$ achieves its minimum at
$\rho=-\frac{\alpha^2}{\sqrt{D_e}}$. Since $\rho \in [-\alpha,\alpha]$ we
conclude that in the case where $\alpha^2 \ge D_e$, $\hrho_{\bu \by} < T$ implies
that $\rho < \sqrt{D_e}(T^2-1)+T \sqrt{\alpha^2-D(1-T^2)}$ ($\hrho_{\bu
\by}(\rho)$ is monotonically increasing in $\rho$, and $\hrho_{\bu
\by}(-\alpha)=-1$). If $(1-T^2)D \le \alpha^2 < D_e$, $\hrho_{\bu \by} < T$
implies that
$$
\sqrt{D_e}(T^2-1)-T \sqrt{\alpha^2-D(1-T^2)}  \le \rho \le
\sqrt{D_e}(T^2-1)+T \sqrt{\alpha^2-D(1-T^2)} \;\;.
$$
Otherwise, for $\alpha^2 < (1-T^2)D_e$,
$\hrho_{\bu \by} \ge T$ for all $\rho \in [-\alpha,\alpha]$. Define
\be
\Psi_1(r) &\eqde& \arccos \left[\frac{\sqrt{D_e}(T^2-1)+T \sqrt{r-D(1-T^2)}}{\sqrt{r}}\right] \;\;,\\
\Psi_2(r) &\eqde& \arccos \left[\frac{\sqrt{D_e}(T^2-1)-T
\sqrt{r-D(1-T^2)}}{\sqrt{r}}\right] \;\;.
\ee
We need to pay attention to the point $r_0=\frac{D_e(1-T^2)}{T^2}$ in which
$\Psi_1(r_0)=\pi/2$. Beyond that point ($r>r_0$), the probability of false-negative
given $\alpha^2=r$ goes to one as $n$ tends to infinity.
Therefore, the false-negative probability can be written as follows: In the case
where $\frac{1-T^2}{T^2}>1$ (or $\lambda < \frac{1}{2} \ln (2)$)
\be
\begin{split}
\label{int1}
P_{fn}^{(n)} =  \frac{(1/2)^{\frac{n}{2}} n^{\frac{n-2}{2}}}{\Gamma(n/2)}
\Bigg[\int_{D_e(1-T^2)}^{D_e} \frac{\sigma^2}{r} \left(e^{n \ln
{\sin\big(\Psi_1(r)\big)}}-e^{n \ln {\sin\big(\Psi_2(r)\big)}}\right)
        e^{- \frac{nr}{2\sigma^2}} e^{\frac{n}{2} \ln(r/\sigma^2)}dr \\
    + \int_{D_e}^{\frac{D_e(1-T^2)}{T^2}}
      \frac{\sigma^2}{r} e^{n \ln {\sin\big(\Psi_1(r)\big)}} e^{- \frac{nr}{2\sigma^2}} e^{\frac{n}{2} \ln(r/\sigma^2)}dr\\
    + \int^{\infty}_{\frac{D_e(1-T^2)}{T^2}}
        \frac{\sigma^2}{r}\left(1- e^{n \ln {\sin\big(\Psi_1(r)\big)}}\right) e^{- \frac{nr}{2\sigma^2}} e^{\frac{n}{2} \ln(r/\sigma^2)} dr
        \Bigg] \;\;.
\end{split}
\ee
The first integral in \eqref{int1} represents the false-negative probability when
both $\Psi_1(r)$ and $\Psi_2(r)$ are greater than $\pi/2$. In this case, we need to
subtract the areas of two caps, i.e.,
$\frac{A_n(\pi-\Psi_1(r))-A_n(\pi-\Psi_2(r))}{A_n(\pi)}$. The second integral in
\eqref{int1} stems from the fact that for $r \ge D_e$ the false-negative probability
(given $\alpha^2=r$) equals to $\frac{A_n(\pi-\Psi_1(r))}{A(\pi)}$. The last
integral in \eqref{int1} stems from the fact that the false-negative probability
(given $\alpha^2=r$) equals to $1-\frac{A(\Psi_1(r))}{A(\pi)}$. In a similar way, in
the case where $\frac{1-T^2}{T^2} \le 1$ (or $\lambda \ge \frac{1}{2} \ln (2)$)
\be
\begin{split}
P_{fn}^{(n)} =  \frac{(1/2)^{\frac{n}{2}} n^{\frac{n-2}{2}}}{\Gamma(n/2)}
\Bigg[\int_{D_e(1-T^2)}^{\frac{D_e(1-T^2)}{T^2}} \frac{\sigma^2}{r} \left(e^{n \ln
{\sin\big(\Psi_1(r)\big)}}-e^{n \ln {\sin\big(\Psi_2(r)\big)}}\right)
        e^{- \frac{nr}{2\sigma^2}} e^{\frac{n}{2} \ln(r/\sigma^2)}dr \\
    + \int_{\frac{D_e(1-T^2)}{T^2}}^{D_e}
      \frac{\sigma^2}{r} \left(1-e^{n \ln {\sin\big(\Psi_1(r)\big)}}-e^{n \ln {\sin\big(\Psi_2(r)\big)}}\right)
        e^{- \frac{nr}{2\sigma^2}} e^{\frac{n}{2} \ln(r/\sigma^2)}dr \\
    + \int^{\infty}_{D_e}
        \frac{\sigma^2}{r} \left(1- e^{n \ln {\sin\big(\Psi_1(r)\big)}}\right) e^{- \frac{nr}{2\sigma^2}} e^{\frac{n}{2} \ln(r/\sigma^2)} dr
        \Bigg] \;\;.
\end{split}
\ee
Since we are interested in the exponential decay rate (to the first order), the
slowest exponent dictates the overall exponential behavior. Therefore, the fact that
$\sin\big(\Psi_1(r)\big) > \sin\big(\Psi_2(r)\big)$ for $D_e(1-T^2) \le r \le
D(1-T^2)/T^2 $ implies that
\be
\begin{split} P_{fn} \doteq  \frac{(1/2)^{\frac{n}{2}} n^{\frac{n-2}{2}}}{\Gamma(n/2)}
\Bigg[\int_{D_e(1-T^2)}^{\frac{D_e(1-T^2)}{T^2}} \frac{\sigma^2}{r} e^{n \ln
{\sin\big(\Psi_1(r)\big)}}
        e^{- \frac{nr}{2\sigma^2}} e^{\frac{n}{2} \ln(r/\sigma^2)}dr \\
    + \int^{\infty}_{\frac{D_e(1-T^2)}{T^2}}
        \frac{\sigma^2}{r} e^{- \frac{nr}{2\sigma^2}} e^{\frac{n}{2} \ln(r/\sigma^2)} dr
        \Bigg] \;.
\end{split}
\ee
Again, using the Laplace's method for integrals \cite[Ch.4]{DeBruijn70} we can
conclude that
\be
E_{fn}^{ae}(T,D_e)= \min \Big\{E_1(T,D_e),E_2(T,D_e) \Big\} \;,
\ee
where,
\be
E_1(T,D_e)& = &\min_{D_e(1-T^2) < r \le \frac{D_e(1-T^2)}{T^2}} \frac{1}{2} \Bigg[\frac{r}{\sigma^2}-\ln\left(\frac{r}{\sigma^2}\right)-2\ln \sin \big(\Psi_1(r)\big)-1 \Bigg]\;, \\
E_2(T,D_e)& = &\min_{r > \frac{D_e(1-T^2)}{T^2}} \frac{1}{2}
\Bigg[\frac{r}{\sigma^2}-\ln\left(\frac{r}{\sigma^2}\right)-1 \Bigg] \;.
\ee
$E_2(T,D_e)$ is given by
\be
E_2(T,D_e)=\left\{ \begin{array}{lll} 0 & , & \frac{D_e(1-T^2)}{T^2} \le \sigma^2 \\
                          \frac{1}{2}\left[\frac{D_e(1-T^2)}{T^2 \sigma^2}-\ln\left(\frac{D_e(1-T^2)}{T^2 \sigma^2}\right)-1 \right] & , & \textrm{else} \end{array}  \right. \;\;.
\ee
Since $T^2=1-e^{-2 \lambda}$, then $E_2(\lambda,D_e)=E_{fn}^{se}(\lambda,D_e)$ and
therefore $E_{fn}^{ae}(\lambda,D_e) \le E_{fn}^{se}(\lambda,D_e)$. Our next step will be
to prove that $E_1(T,D_e) < E_2(T,D_e)$ when $\frac{D_e(1-T^2)}{T^2}>\sigma^2$ (otherwise,
$E_{fn}^{ae}(T,D_e)=0$). Define
\be
f(r)=\frac{r}{2\sigma^2}-\frac{1}{2}\ln\left(\frac{r}{\sigma^2}\right)-\ln \sin
\big(\Psi_1(r)\big)-\frac{1}{2}
\ee
$f(r)$ is a continuous, non-negative function in the range $D_e(1-T^2) < r \le
\frac{D_e(1-T^2)}{T^2}$. Clearly,
\be
E_1(T,D_e) \le f\left( \frac{D_e(1-T^2)}{T^2} \right) = E_2(T,D_e).
\ee
In addition, $f'(r)$ is continuous in the above range. It can easily be shown that
\be
f'\left(\frac{D_e(1-T^2)}{T^2} \right)= \frac{1}{2}\left[1-\frac{T^2
\sigma^2}{D_e(1-T^2)} \right] > 0
\ee
hence, $f(r)$ is monotonically increasing in small neighborhood of
$\frac{D_e(1-T^2)}{T^2}$, and therefore $E_1(T,D_e) < E_2(T,D_e)$. This fact leads to the
conclusion that $E_{fn}^{ae}(\lambda,D_e) < E_{fn}^{se}(\lambda,D_e)$. The exact value
of $E_1(T,D_e)$ is cumbersome and therefore will not be presented.
\end{proof}

%-------------------------------------------------------------------------------------------------
%-------------------------------------------------------------------------------------------------
%% Bib:
%\pagebreak
%GATHER{../watermarking.bib}   % For Gather Purpose Only (WinEdt command)
\bibliographystyle{IEEEtran}
\bibliography{IEEEabrv,../watermarking}
%-------------------------------------------------------------------------------------------------
%-------------------------------------------------------------------------------------------------
\end{document}